\documentclass[journal]{IEEEtran}
\usepackage[T1]{fontenc}
\usepackage[latin9]{inputenc}
\usepackage{amsthm}
\usepackage{amsmath} 
\usepackage{graphicx} 
\usepackage{color} 
\usepackage{cite}
\usepackage{algpseudocode}
\usepackage{algorithm}
\usepackage{amssymb}
\usepackage{subfigure}
\usepackage{esint}
\usepackage{epstopdf}
\usepackage{multirow}
\usepackage{makecell}
\usepackage{float}
\usepackage{lipsum}
\usepackage{balance}
\usepackage{ulem}
\usepackage{bbm}
\usepackage{comment}
\usepackage{enumitem}

\newtheorem{definition}{Definition}
\newtheorem{remark}{Remark}

\theoremstyle{plain}
\newtheorem{theorem}{Theorem}
\newtheorem{lemma}{Lemma}

\IEEEoverridecommandlockouts

\algnewcommand\algorithmicforeach{\textbf{for each}}
\algdef{S}[FOR]{ForEach}[1]{\algorithmicforeach\ #1\ \algorithmicdo}

\definecolor{green}{rgb}{0.13, 0.55, 0.13}
\definecolor{brown}{rgb}{0.6, 0.2, 0.0}
\def\thickhline{\noalign{\hrule height1.2pt}}

\allowdisplaybreaks

\begin{document}

\title{A Markov Chain Monte Carlo Method for Efficient Finite-Length LDPC Code Design}

\author{
    \IEEEauthorblockN{Ata Tanr{\i}kulu, Mete Y{\i}ld{\i}r{\i}m, and Ahmed Hareedy, \IEEEmembership{Member, IEEE}}
    \thanks{This work was supported in part by the T\"{U}B\.{I}TAK 2232-B International Fellowship for Early Stage Researchers.
   
    Ata Tanr{\i}kulu, Mete Y{\i}ld{\i}r{\i}m, and Ahmed Hareedy are with the Department of Electrical and Electronics Engineering,
Middle East Technical University (METU), 06800 Ankara, Turkey (e-mail: ata.tanrikulu@metu.edu.tr; mete.yildirim@metu.edu.tr; ahareedy@metu.edu.tr).}
}
\maketitle

\vspace{-5.0em}
\begin{abstract}

Low-density parity-check (LDPC) codes are among the most prominent error-correction schemes in today's data-driven world. They find application to fortify various modern storage, communication, and computing systems. Protograph-based (PB) LDPC codes offer many degrees of freedom in the code design and enable fast encoding and decoding. In particular, spatially-coupled (SC) and multi-dimensional (MD) circulant-based codes are PB-LDPC codes with excellent performance. Efficient finite-length (FL) algorithms are required in order to effectively exploit the available degrees of freedom offered by SC partitioning, lifting, and MD relocations. In this paper, we propose a novel Markov chain Monte Carlo (MCMC or MC\textsuperscript{2}) method to perform this FL optimization, addressing the removal of short cycles. While we focus on partitioning and lifting of SC codes, our MC\textsuperscript{2} approach can effectively work for other procedures and/or other code designs. While iterating, we draw samples from a defined distribution where the probability decreases as the number of short cycles from the previous iteration increases. We analyze our MC\textsuperscript{2} method theoretically as we prove the invariance of the Markov chain where each state represents a possible partitioning or lifting arrangement, i.e., sample, that has a specific probability. Via our simulations, we then fit the distribution of the number of cycles resulting from a given arrangement on a Gaussian distribution. By analyzing the mean, we derive estimates for cycle counts that are close to the actual counts. Furthermore, we derive the order of the expected number of iterations required by our MC\textsuperscript{2} approach to reach a local minimum as well as the size of the Markov chain recurrent class through approximating the probability of getting arbitrarily close to the local minimum. Our approach is compatible with code design techniques based on gradient-descent. Experimental results show that our MC\textsuperscript{2} method generates SC codes with remarkably fewer short cycles and substantial gains in error/erasure-rate performance compared with the current state-of-the-art. Moreover, to reach the same number of cycles, our MC\textsuperscript{2} method requires orders of magnitude less overall time compared with the available literature methods.


\end{abstract}


\section{Introduction} \label{sec:intro}

Since their introduction by Gallager in 1963 \cite{gal_th} then their reinvention by MacKay in 1999 \cite{mackay_cod}, low-density parity-check (LDPC) codes have been increasingly gaining ground as the error-correction scheme of choice in various modern systems. Today, LDPC modules exist in various data storage, data transmission, as well as digital communication systems \cite{hareedy-flash}, \cite{cai}, \cite{homa}, \cite{ahh_jsac}, \cite{declercq} because of their capacity-approaching performance and their practical message-passing decoding \cite{mackay_cod}. Protograph-based LDPC codes are designed from a base matrix (graph), called a protograph matrix (protograph), that is lifted via circulant permutation matrices (CPMs) or, in short, circulants \cite{fossorier}, \cite{costello}. In addition to the degrees of freedom they offer the code designer via the circulant powers, these codes also support fast and efficient encoding and decoding procedures. In 2004, Fossorier introduced algebraic conditions on the circulant powers, which specify the circular shift of each CPM diagonal elements, to maintain or remove a protograph cycle after lifting \cite{fossorier}. Throughout this paper, we interchangeably use the terms ``matrix'' and its corresponding ``graph'' when the context is clear.

Recent innovations in LDPC code design and decoding led to the introduction of spatially-coupled (SC) \cite{costello}, \cite{oocpo} and multi-dimensional (MD) codes \cite{schmalen}, \cite{ahh_md}. SC codes are designed by partitioning an underlying block matrix and coupling the components \cite{battaglioni}, \cite{banihashemi}. MD codes are designed by relocating entries from one-dimensional LDPC code copies to connect them \cite{ahh_md}. SC codes offer capacity-achievability \cite{kudekar}, \cite{kudekar2}, low-latency if windowed decoding is adopted \cite{siegel}, \cite{kwak}, as well as additional degrees of freedom for finite-length code design coming from partitioning. MD codes offer even further flexibility and they can mitigate (multi-dimensional) channel non-uniformity \cite{schmalen}. To better exploit the degrees of freedom for SC and MD codes, we recently introduced a probabilistic approach that is based on the gradient-descent (GD) algorithm to efficiently design locally-optimal codes with respect to minimizing the multiplicities of detrimental cycles/objects in their graph representations \cite{GRADE}, \cite{reins_gdmd}. Detrimental cycles/objects here are small ones that are either dominant absorbing sets \cite{absorbing} or common subgraphs of dominant absorbing sets \cite{channel_aware}.

Protograph-based SC (MD) code design procedure consists of two (three) stages, partitioning and lifting (partitioning, lifting, and relocations) \cite{oocpo}, \cite{ahh_md}, \cite{channel_aware}, \cite{rosnes}, \cite{battaglioni-globe}, \cite{battaglioni_mitchell}. With or without probabilistic approaches in the design, finite-length (FL) algorithmic optimization is necessary for all the aforementioned stages. For example, various circulant-power optimizers (CPOs) were presented to perform the stage of lifting  \cite{oocpo}, \cite{channel_aware}. There are two metrics to gauge the efficiency of an FL algorithmic optimizer (AO). The first is performance, measured by how much it can reduce the number of target cycles/objects. The second is execution time, measured by how many iterations required to reach some count. As the degrees of freedom in the SC or MD code design increase, the complexity of existing FL-AOs grows rapidly.

Markov chain Monte Carlo (MCMC) methods are effective methods to draw samples from a specific probability distribution based on some metric \cite{mackay_it}, \cite{metropolis_mc}, \cite{hastings_mh}, and they can be useful in discrete optimization \cite{elperin_monte}\cite{li_gibbs}. In this paper, we offer the first MCMC framework, which we refer to as our MC\textsuperscript{2} method, that effectively performs FL optimization to design high performance LDPC codes. In particular, the proposed MC\textsuperscript{2} method can reach locally-optimal parameters for SC partitioning, MD relocations, and circulant-power lifting in remarkably less computational time compared with the available state-of-the-art. Optimality here is with respect to minimizing the number of short cycles or small graphical objects (here, these are also cycles with specific properties, and they are common subgraphs of dominant detrimental objects \cite{ahh_md}, \cite{channel_aware}). Without loss of method generality, we focus on the design of SC and MD-SC codes in this work.

In our MC\textsuperscript{2} method, the samples drawn represent partitioning or lifting arrangements, and the distribution is constructed such that the probability is inversely proportional to the total number of cycles of interest. The core contributions of the proposed MC\textsuperscript{2} method are summarized as follows:

\begin{itemize} 
    \item We provide a theoretical analysis proving the invariance and aperiodicity properties of the Markov chain defined by our MC\textsuperscript{2} sampling process.
    
    \item Through simulations, we collect data regarding the number of cycles and fit the resulting histogram to a Gaussian distribution.
    
    \item We derive a differential equation where the rate of change of the Gaussian mean of the cycle count is expressed as a function of its variance.
    
    \item Using the fitted model and differential equation, we estimate the number of cycles, obtaining values close to the actual observed cycle counts after optimization.
    
    \item We approximate the probability that the cycle count lies within an arbitrary $\epsilon$ of the local optimum, and use this probability to estimate:
    \begin{itemize}
        \item the order of the number of iterations required to reach the local optimum, and
        \item the number of states in the recurrent class of the Markov chain.
    \end{itemize}
    
    \item We offer experimental results that demonstrate:
    \begin{itemize}
        \item significant reductions in the number of cycles of lengths $6$ and $8$ compared with existing FL-AOs,
        \item orders-of-magnitude improvements in computational time to reach comparable or better cycle counts, and
        \item frame error/erasure-rate performance gains of up to $1.61$ orders of magnitude.
    \end{itemize}

    \item We show that when the local optimum search space is reduced using the GD algorithm~\cite{GRADE}, our MC\textsuperscript{2} method achieves even greater performance improvements.

    \item With this method, we offer code designers an FL optimization algorithm that can outperform other heuristic FL algorithms~\cite{oocpo, GRADE}, while also enabling the design of SC codes with high memory, which is beyond the practical limits of optimal methods~\cite{oocpo}.

\end{itemize}

The rest of the paper is organized as follows. Preliminaries about SC code construction and MCMC methods are presented in Section~\ref{sec:prelim}. Our MC\textsuperscript{2} probabilistic optimizer (the framework and the algorithm) is introduced in Section~\ref{sec:op_design}. Data-fitting-based estimation of the cycle count, iteration number, and recurrent class size is discussed in Section~\ref{sec:fitting}. Numerical results on the cycle counts, computational times, estimation results, and error/erasure-rate performance for various codes are presented in Section~\ref{sec:numeric}. Finally, the paper is concluded in Section~\ref{sec:conclude}.

\section{Preliminaries} \label{sec:prelim}

\subsection{Spatially-Coupled Codes}

Let $\mathbf{H}_{\textup{SC}}$ in (\ref{sc matrix}) be the parity-check matrix of a circulant-based (CB) spatially-coupled (SC) code with parameters $\left(\gamma, \kappa, z, L, m\right)$, where $\gamma \geq 3$ and $\kappa > \gamma$ are the column and row weights of the underlying block code, respectively. The SC coupling length and memory are $L$ and $m$, respectively.

\vspace{-0.7em}\begin{gather} \label{sc matrix}
\mathbf{H}_{\textup{SC}} =
\begin{bmatrix}
\mathbf{H}_0 & \mathbf{0}  &  & & & \mathbf{0} \vspace{-0.5em}\\
\mathbf{H}_1 & \mathbf{H}_0 & &  &  & \vdots \vspace{-0.5em}\\
\vdots & \mathbf{H}_1 & \ddots &  &  & \vdots \vspace{-0.5em}\\
\mathbf{H}_m & \vdots & \ddots & \ddots & & \mathbf{0} \vspace{-0.5em}\\
\mathbf{0} & \mathbf{H}_m & \ddots & \ddots & \ddots  & \mathbf{H}_0 \vspace{-0.5em}\\
\vdots & \mathbf{0} & & \ddots & \ddots & \mathbf{H}_1 \vspace{-0.5em}\\
\vdots & \vdots &  & &  \ddots & \vdots \vspace{-0.0em}\\
\mathbf{0} & \mathbf{0}& &  &  & \mathbf{H}_m
\end{bmatrix}.
\end{gather}

$\mathbf{H}_{\textup{SC}}$ is obtained from a binary matrix $\mathbf{H}^{\textup{g}}_{\textup{SC}}$ by replacing each nonzero (zero) entry in the latter by a circulant (zero) $z \times z$ matrix, $z \in \mathbb{N}$. The notation ``$\textup{g}$'' in the superscript of any matrix refers to its protograph matrix, where each circulant matrix is replaced by $1$ and each zero matrix is replaced by $0$. $\mathbf{H}^{\textup{g}}_{\textup{SC}}$ is called the protograph matrix of the SC code, and it has the same convolutional structure in (\ref{sc matrix}) composed of $L$ replicas $\mathbf{R}_{i}^{\textup{g}}$ of size $\left(m+L\right)\gamma \times \kappa$ each, where 
\[
\mathbf{R}_i^{\textup{g}} = \left[(\mathbf{0}_{\gamma i  \times \kappa })^{\textup{T}} \textup{} (\mathbf{H}_0^{\textup{g}})^{\textup{T}} \textup{ } (\mathbf{H}_1^{\textup{g}})^{\textup{T}} \textup{ } \dots \textup{ } (\mathbf{H}_m^{\textup{g}})^{\textup{T}} \textup{ } (\mathbf{0}_{\gamma (L-1-i)  \times \kappa })^{\textup{T}}\right]^{\textup{T}},
\]
 and $\mathbf{H}_y^{\textup{g}}$'s are all of size $\gamma  \times \kappa$, and are pairwise disjoint. The sum $\mathbf{{H}^{\textup{g}}} = \sum_{y=0}^m \mathbf{H}_y^{\textup{g}}$ is called the base matrix, which is the protograph matrix of the underlying block matrix $\mathbf{H} = \sum_{y=0}^m \mathbf{H}_y$. In this paper, $\mathbf{{H}^{\textup{g}}}$ is taken to be all-one matrix, and the SC codes have quasi-cyclic (QC) structure. 
\par
Here, the $\gamma \times \kappa$ matrix $\mathbf{P}$ whose entry at $(i,j)$ is $\mathbf{P}(i,j)=a \in \{0,1,\dots,m\}$ when $\mathbf{H}_a^{\textup{g}}=1$ at that entry is called the partitioning matrix. The $\gamma \times \kappa$ matrix $\mathbf{L}$ with $\mathbf{L}(i,j)=f_{i,j} \in \{0,1,\dots,z-1\}$, where $\sigma^{f_{i,j}}$ is the circulant in $\mathbf{H}$ lifted from the entry $\mathbf{{H}^{\textup{g}}}(i,j)$, is called the lifting matrix. Here, $\sigma$ is the $z \times z$ identity matrix with its columns cyclically shifted one unit to the left. Observe that $m+1$ is the number of component matrices $\mathbf{H}_y^{\textup{g}}$ (or $\mathbf{H}_y$), and $L$ is the number of replicas in $\mathbf{H}^{\textup{g}}_{\textup{SC}}$ (or $\mathbf{H}_{\textup{SC}}$).

A special class of SC codes is the class of topologically-coupled (TC) codes \cite{GRADE}, characterized by pseudo-memory features. In TC codes, a deliberate design choice is to set a subset of component matrices to be entirely zero. The terminology ``topologically-couple'' stems from the topological degrees of freedom available to the code designer in selecting the non-zero component matrices.

A cycle-$2g$ candidate in $\mathbf{H}^{\text{g}}_{\text{SC}}$ is a way of traversing a structure to generate cycles of length $2g$ after lifting \cite{channel_aware}. Similarly, we define object candidates as ways of traversing structures in the protograph to generate objects after lifting. The same concepts also apply to $\mathbf{H}^{\text{g}}$ for partitioning \cite{battaglioni}. In an SC code, each cycle in the Tanner graph of $\mathbf{H}_{\text{SC}}$ corresponds to a cycle candidate in the protograph matrix $\mathbf{H}^{\text{g}}_{\text{SC}}$, and each cycle candidate in $\mathbf{H}^{\text{g}}_{\text{SC}}$ corresponds to a cycle candidate $\mathcal{C}$ in the base matrix $\mathbf{H}^{\text{g}}$.

Lemma~\ref{lemma:cycle_condition} specifies a necessary and sufficient condition for a cycle candidate in $\mathbf{H}^{\textup{g}}$ to become a cycle candidate in the SC protograph and then a cycle in the final SC Tanner graph.

\begin{lemma} \label{lemma:cycle_condition}
Consider an SC code constructed as illustrated above. Let $ \mathcal{C} $ be a cycle-$2g$ candidate in the base matrix, where $ g \in \mathbb{N}$ and $g \geq 2 $. Denote $ \mathcal{C} $ by $ \left(i_1,j_1, i_2,j_2, \ldots, i_g, j_g\right) $, where $ \left(i_k, j_k\right) $ and $ \left(i_k, j_{k+1}\right) $, $ 1 \leq k \leq g $, $ j_{g+1} = j_1 $, are nodes (entries) of $ \mathcal{C} $ in $\mathbf{H}^{\textup{g}}$. The partitioning and lifting matrices are $ \mathbf{P} $ and $ \mathbf{L} $, respectively. Then $ \mathcal{C} $ becomes a cycle-$2g$ candidate in the SC protograph, i.e., remains active, if and only if the following condition follows \cite{battaglioni}:
\begin{equation} \label{eq:partition_condition}
\sum_{k=1}^g \mathbf{P}(i_k, j_k) = \sum_{k=1}^g \mathbf{P}(i_k, j_{k+1}). 
\end{equation}
This cycle candidate becomes a cycle-$2g$ in the SC Tanner graph, i.e., remains active, if and only if \cite{fossorier}:
\begin{equation} \label{eq:lifting_condition}
\sum_{k=1}^g \mathbf{L}(i_k, j_k) \equiv \sum_{k=1}^g \mathbf{L}(i_k, j_{k+1}) \pmod{z}.
\end{equation}
\end{lemma}

\subsection{Markov Chain Monte Carlo Methods}

Markov chains are discrete-time stochastic processes that probabilistically undergo transitions from one state to another within a state space. They are characterized by the property that given the current state, the future state becomes independent of the past states. Monte Carlo methods are computational algorithms that rely on repeated random sampling to obtain numerical results. They are particularly useful in calculating integrals, simulating systems with inherent randomness, and emulating high-dimensional systems.

Markov chain Monte Carlo (MCMC) methods integrate the concepts of Markov chains and Monte Carlo techniques. These methods are used to sample from sophisticated probability distributions, which is particularly useful in high-dimensional spaces where direct sampling is challenging. MCMC methods require two main conditions for success: the invariance of the distribution and the ergodicity of the Markov chain or its recurrent class of interest \cite{mackay_it}. Invariance ensures that if the chain starts from the stationary distribution, it remains in this distribution at every step (iteration). Ergodicity guarantees that the chain can reach any part of the state space in a finite number of steps, which implies that the sampling distribution converges to the stationary distribution as more samples are drawn, i.e., more iterations are performed.

One specific MCMC technique of interest is the Metropolis-Hastings algorithm, which allows sampling from a distribution by generating a sequence of sample values that effectively explore the target distribution \cite{metropolis_mc,hastings_mh}. This method involves suggesting moves in the state space, which are the samples, according to a proposed distribution and accepting or rejecting each move based on the acceptance probability that ensures invariance to appropriately represent the target distribution.

Another technique is Gibbs sampling, a special case of the Metropolis-Hastings algorithm, particularly used when the conditional probabilities of each variable, given all other variables, are easy to compute \cite{geman_gibbs}. These variables constitute the state. In Gibbs sampling, we access each variable to sample from its conditional distribution:
\begin{equation}
x_i^{\left(t+1\right)} \sim P\left(x_i \mid x_1^{\left(t+1\right)}, \ldots, x_{i-1}^{\left(t+1\right)}, x_{i+1}^{\left(t\right)}, \ldots, x_n^{\left(t\right)}\right),
\end{equation}
where $ x_i $ is a component of the state vector $ \mathbf{x} $ of length $n$ and $ t $ denotes the time-step. This method is more feasible when conditional distributions are straightforward to calculate, making it particularly effective for advanced, high-dimensional distributions. Traditionally used to update one variable at a time based on its conditional probabilities relative to other variables, the Gibbs sampler can be adapted to update multiple variables simultaneously if their joint conditional probabilities are computationally tractable.

\section{Probabilistic Optimizer Design} \label{sec:op_design}

In this section, we introduce our optimization problem, which is the general case of short-cycle or common-substructure minimization in the SC code design at the partitioning and lifting stages.\footnote{A common substructure is a common subgraph of various absorbing sets that dominate the error profile of the code when the channel conditions are better \cite{ahh_md}, \cite{channel_aware}.} We also present our proposed Gibbs sampling-based MC\textsuperscript{2} finite-length (FL) optimization method. 

\subsection{Optimization Problem Overview}

We define two separate optimization problems for the partitioning and the lifting stages of SC code design. Our goal is to find optimal or suboptimal partitioning and lifting matrices to reduce the number of detrimental objects of interest. We denote the input or state vector by $\mathbf{x}$, which is obtained by concatenating the rows of either the partitioning or the lifting matrix, depending on the problem under consideration. We define $C(\mathbf{x})$ as the objective function, which is the total number of detrimental objects under this $\mathbf{x}$, and $F$ as the feasible set of the problem.

The reason these two stages of code design are separated is that joint design does not consistently result in better solutions compared with separate design, and the former significantly increases the complexity of the design process~\cite{oocpo}, \cite{rosnes}.

At the partitioning stage, we focus on minimizing the count of short cycle candidates in $\mathbf{H}^{\textup{g}}_{\textup{SC}}$ by focusing on short cycle candidates in the base matrix. The optimization problem is formulated as follows:
\begin{equation}
\underset{\mathbf{x}}{\text{minimize}} \;
C(\mathbf{x}) = w_4 \, C_4(\mathbf{x}) + w_6 \, C_6(\mathbf{x}) + w_8 \, C_8(\mathbf{x}),
\end{equation}
where $\left\{0, 1, \ldots, m\right\}^{\gamma\kappa}$ is the domain and the feasible set. $C_4(\mathbf{x})$, $ C_6(\mathbf{x})$, and $C_8(\mathbf{x})$ represent the number of cycle-$4$, cycle-$6$, and cycle-$8$ candidates in the base matrix $\mathbf{H}^{\textup{g}}$ that remain active given the partitioning specified by $\mathbf{x}$, and $w_4$, $w_6$, and $w_8$ are their corresponding weights. We incorporate the probabilistic framework of \cite{GRADE}, which is based on gradient-descent (GD), to approach this problem as we use the output from the GD distributor to guide partitioning. The MC\textsuperscript{2} FL optimizer is initialized using this output.

For the lifting stage, the objective is to minimize the occurrence of short cycles or common substructures of interest while keeping smaller cycles and/or low-weight codewords inactive in the final Tanner graph of the SC code. The problem is defined as:
\begin{align}
&\underset{\mathbf{x}}{\text{minimize}} \; C(\mathbf{x}) \nonumber \\  
&\text{subject to no active objects in } \mathcal{L}',
\end{align}
where $\left\{0, 1, \ldots, z-1\right\}^{\gamma\kappa}$ is the domain and $C(\mathbf{x})$ is the number of objects in $\mathbf{H}^{\textup{g}}_{\textup{SC}}$ that remain active after lifting is done based on $\mathbf{x}$. $\mathcal{L}'$ is the list of smaller objects that we wish to keep inactive, and the feasible set $F$ represents the subset of the domain where each element ensures that all objects in the set $\mathcal{L}'$ remain inactive.

\subsection{Theoretical Framework}

In our derivations, we focus on a normalized version of the objective function defined as:
\begin{equation}
\overline{C}(\mathbf{x}) = \frac{C(\mathbf{x})}{\alpha},
\end{equation}
where $\alpha$ is the normalizing factor, and it is equal to the maximum value that the objective function could take.  From this point forward, we will refer to the normalized version as the objective function for brevity.

\begin{definition} [Main distribution]
The joint probability distribution 
\begin{equation}
P(\mathbf{x}) = \frac{P^*(\mathbf{x})}{Z(\beta)}
\end{equation}
of the input vector $\mathbf{x}$, where
\begin{equation}
P^*(\mathbf{x}) = 
\begin{cases}
e^{-\beta \overline{C}(\mathbf{x})}, & \text{if } \mathbf{x} \in F, \\
0, & \text{otherwise,}
\end{cases}
\end{equation}
is defined as the main distribution. Here, $Z(\beta)$ is the normalizing factor
\begin{equation}
Z(\beta) = \sum_{\mathbf{x} \in F} e^{-\beta \overline{C}(\mathbf{x})}
\end{equation}
and $\beta \in \mathbb{R}^{\geq 0}$ is a hyper-parameter.
\end{definition}

The probability value of the main distribution is inversely proportional to the objective function, which results in highest probabilities being assigned to near-optimal and optimal solutions of our problem. While computing $P^*(\mathbf{x})$ for any $\mathbf{x} \in F$ requires only a single evaluation, computing $Z(\beta)$ necessitates $|F|$ evaluations of the objective function. This scales as $O\left(N^{\gamma \kappa}\right)$, where $N$ represents the number of possible values for each entry of $\mathbf{x}$. Thus, evaluating $P(\mathbf{x})$ is computationally challenging. Our primary goals are to generate samples from this computing-wise challenging distribution using Monte Carlo methods and to exploit the proportionality between optimal results and their corresponding high probabilities so that we can build an efficient MC\textsuperscript{2} FL optimizer.

The conditional probability of a subset of entries of $\mathbf{x}$ given the remaining entries, under the main distribution, is:

\begin{equation}
P\left(\mathbf{x}_\nu \mid \mathbf{x}_{\backslash \nu}\right) = \frac{P(\mathbf{x})}{\sum_{\mathbf{x}': \ \mathbf{x}'_{\backslash \nu} = \mathbf{x}_{\backslash \nu} } \hspace{-0.3em} P(\mathbf{x}')} = \frac{P^*(\mathbf{x})}{\sum_{\mathbf{x}': \ \mathbf{x}'_{\backslash \nu} = \mathbf{x}_{\backslash \nu} } \hspace{-0.3em} P^*(\mathbf{x}')}.
\label{eq:conditional_probability}
\end{equation}
Here, $\nu$ denotes the set of selected indices, $\backslash \nu$ refers to the remaining indices, and $\mathbf{x}_{\nu}$ represents the subset of entries of $\mathbf{x}$ indexed by $\nu$. Evaluating the probabilities for this conditional distribution requires $O\left(N^d\right)$ evaluations of the objective function, where $d = |\nu|$. For small $d$, this is computationally feasible, enabling the use of Gibbs sampling.

In our optimization algorithm, we construct a Markov chain where each state corresponds to an $\mathbf{x} \in F$, and transitions are based on conditional probabilities from Equation~(\ref{eq:conditional_probability}). A list of $d$-tuples of indices is maintained whose length is same as that of the input vector $\mathbf{x}$. Each tuple on this list involves one of the indices, say $\zeta$, from the set $\left\{0,1,\ldots, \left|\mathbf{x}\right|-1\right\}$ corresponding to the entries of $\mathbf{x}$ and the $d-1$ indices of the most correlated entries to the one indexed by $\zeta$. Correlation is determined by the number of objects of interest that share these entries or, for a weighted objective function, the weighted sum of these numbers for different object types.

The algorithm iterates over the list of $d$-tuples, evaluating  
the conditional probability mass function (PMF) as the transition probability for the current index $d$-tuple using (\ref{eq:conditional_probability}). The transition probability from the current state $\mathbf{x}$ to the next state $\mathbf{x}'$, given that the current (selected) index $d$-tuple is $\nu$, is denoted by $T\left(\mathbf{x}';\mathbf{x} \mid \nu\right)$ and has the below relation:
\begin{equation}
T\left(\mathbf{x}';\mathbf{x} \mid \nu\right) = 
\begin{cases}
P\left(\mathbf{x}_\nu' \mid \mathbf{x}_{\backslash \nu}\right), & \text{if } \mathbf{x}_{\backslash \nu} = \mathbf{x}_{\backslash \nu}', \\
0, & \text{otherwise.}
\end{cases}
\end{equation} New states are reached, i.e., new samples are drawn, based on these transition probabilities, and the minimum objective function value as well as the corresponding input vector are updated whenever a better solution is found during iterative evaluations. To avoid repeated patterns, the list of $d$-tuples is randomly shuffled after each pass of the list.

\begin{lemma}\label{lemma:aperiodic_invar} $P(\mathbf{x})$ is an invariant distribution of the Markov chain constructed by the described sampling method above, and this Markov chain is aperiodic.
\end{lemma}

\begin{proof} Let $\mathbf{x}$ and $\mathbf{x'}$ be two states of the Markov chain. Here, $t$ is defined as the tuple of indices at which the entries of these two vectors are different, and $S$ is the list of $d$-tuples indexing entries that are updated during Gibbs sampling. 
We introduce $T\left(\mathbf{x}';\mathbf{x}\right)$ as the overall transition probability from $\mathbf{x}$ to  $\mathbf{x'}$ and $S^{\{t\}}$ as a sub-list of $S$ defined as $S^{\{t\}} = \left\{ t'\mid t\subseteq t' , t'\in S \right\}$, where $\vert t' \vert = d$. We can write transition probabilities as
\begin{align}
T\left(\mathbf{x}';\mathbf{x}\right) &= \sum_{\nu \in S} P(\nu) T\left(\mathbf{x}';\mathbf{x}\mid\nu\right) \nonumber \\ &= \sum_{\nu \in S^{\{t\}}} P(\nu) P\left(\mathbf{x}'_{\nu} \mid \mathbf{x}_{\backslash \nu}\right),
\label{eq:transition_probability}
\end{align}
where $P(\nu)$ is the probability that the $d$-tuple $\nu$ is selected as the tuple of indices of entries to be updated. Also, since $\mathbf{x}$ and $\mathbf{x'}$ only differ at indices that are elements of tuple $t$, $\forall\nu \in S^{\{t\}}, \, \mathbf{x}_{\backslash \nu} = \mathbf{x}'_{\backslash \nu}$, leading to $T\left(\mathbf{x}';\mathbf{x} \mid \nu\right) = P\left(\mathbf{x}'_{\nu} \mid \mathbf{x}_{\backslash \nu}\right) \neq 0$. We can then show that
\begin{align}
&T\left(\mathbf{x}';\mathbf{x}\right) P(\mathbf{x}) = \sum_{\nu \in S^{\{t\}}} P(\nu)  P\left(\mathbf{x}'_{\nu} \mid \mathbf{x}_{\backslash \nu}\right) P(\mathbf{x}) \nonumber \\
&=  \sum_{\nu \in S^{\{t\}}} P(\nu) P\left(\mathbf{x}'_{\nu} \mid \mathbf{x}_{\backslash \nu}\right) P(\mathbf{x}_{\backslash \nu}) P\left(\mathbf{x}_{\nu} \mid \mathbf{x}_{\backslash \nu}\right)\nonumber \\
&=  \sum_{\nu \in S^{\{t\}}} P(\nu) P\left(\mathbf{x}'_{\nu} \mid \mathbf{x}'_{\backslash \nu}\right) P(\mathbf{x}'_{\backslash \nu}) P\left(\mathbf{x}_{\nu} \mid \mathbf{x}'_{\backslash \nu}\right) \nonumber \\
&=  \sum_{\nu \in S^{\{t\}}} P(\nu) P\left(\mathbf{x}_{\nu} \mid \mathbf{x}'_{\backslash \nu}\right) P(\mathbf{x'})\nonumber \\ 
&= T\left(\mathbf{x};\mathbf{x}'\right) P(\mathbf{x'}).
\label{eq:detailed_balance}
\end{align}

Resulting relation~(\ref{eq:detailed_balance}) is known as detailed balance, which implies the invariance of the distribution of $P(\mathbf{x})$ under this Markov chain.

Lastly, from Equations~(\ref{eq:conditional_probability}) and~(\ref{eq:transition_probability}), it can be seen that $T\left(\mathbf{x};\mathbf{x}\right) > 0 $, $\forall\mathbf{x}\in F$. The ability to self-transition makes the Markov chain aperiodic.
\end{proof}

As outlined in Section~\ref{sec:prelim}, for an MCMC method to converge to the desired distribution, the associated Markov chain must be ergodic (has a single recurrent class that is aperiodic) and the given distribution must be invariant from the chain. While Lemma~\ref{lemma:aperiodic_invar} establishes aperiodicity and invariance, the presence of a single recurrent class remains unaddressed. Having said that, we only consider the recurrent class observed during our MC\textsuperscript{2} algorithm as a small chain for which, all convergence conditions are satisfied, leaving further discussion regarding multiple recurrent classes for Section~\ref{sec:numeric}.

\begin{definition} [Main distribution of a recurrent class]
We denote the recurrent class that the MC\textsuperscript{2} algorithm enters by $A$. The main distribution of $A$ is
\begin{equation}
P_A(\mathbf{x}) = 
\begin{cases}
\frac{e^{-\beta \, \overline{C}(\mathbf{x})}}{Z_A(\beta)}, & \text{if } \mathbf{x} \in A, \\
0, & \text{otherwise}.
\end{cases}
\end{equation}
where
\begin{equation} \label{eq:Z_A}
Z_A(\beta) = \sum_{\mathbf{x} \in A} e^{-\beta \, \overline{C}(\mathbf{x})}
\end{equation}
is the normalizing factor for $A$.
\end{definition}

Invariance of $P_{A}(\mathbf{x})$ can be easily shown by altering the proof of Lemma~\ref{lemma:aperiodic_invar} because the definitions of conditional and transition probabilities remain the same, except for operating within $A$ instead of $F$. Since the Markov chain simulation cannot leave this recurrent class once it enters and since this recurrent class is also aperiodic, the steady-state distribution of the main chain on this recurrent class (the small chain), $\pi_A(\mathbf{x})$, converges to $P_{A}(\mathbf{x})$.

\begin{theorem} \label{theorem:convergence} Let $A$ be the recurrent class that the Gibbs sampler enters during the run-time of the algorithm. Then, as the sample count $ t\to\infty $, the difference between the local minimum value that the objective function $ \overline{C}(\mathbf{x})$ could take in this recurrent class $A$ (denoted by $ \overline{C}^*_A $) and the expected value of the objective function under the steady-state distribution of the Markov chain in this recurrent class is upper bounded by $ \frac{1}{\beta} \ln\left(\frac{|A|}{\left|A^*\right|}\right) $, where $ A^* = \left\{ \mathbf{x} \mid \overline{C}(\mathbf{x}) = \overline{C}^*_A, \mathbf{x}\in A \right\} $.
\end{theorem}

\begin{proof}
For any $\mathbf{x}^* \in A^* $ and for the recurrent class of interest
\begin{align}
&\mathbb{E}_A \left\{ e^{-\beta \, (\overline{C}(\mathbf{x}^*) - \overline{C}(\mathbf{x}))} \right\} = \sum_{\mathbf{x} \in A} \pi_A(\mathbf{x}) \, e^{-\beta \, \left(\overline{C}(\mathbf{x}^*) - \overline{C}(\mathbf{x})\right)} \nonumber \\
&= \sum_{\mathbf{x} \in A} \left(\pi_A(\mathbf{x}) \; \frac{\pi_A(\mathbf{x}^*)}{\pi_A(\mathbf{x})}\right) = \pi_A(\mathbf{x}^*) \, |A|.
\end{align}
Since $\pi_A(\mathbf{x}^*)\, |A^*| \leq 1 $, it follows that $ \pi_A(\mathbf{x}^*) \, |A| \leq \frac{|A|}{|A^*|} $. Hence,
\begin{equation}
\mathbb{E}_A \left\{ e^{-\beta \, \left(\overline{C}(\mathbf{x}^*) - \overline{C}(\mathbf{x})\right)} \right\} \leq \frac{|A|}{|A^*|}.
\end{equation}
Using Jensen's inequality for convex functions, we get:
\begin{align}
&\frac{|A|}{|A^*|} \geq \mathbb{E}_A \left\{ e^{-\beta \, \left(\overline{C}(\mathbf{x}^*) - \overline{C}(\mathbf{x})\right)} \right\} = e^{-\beta \, \overline{C}^*_A} \, \mathbb{E}_A \left\{ e^{\beta \, \overline{C}(\mathbf{x})} \right\} \nonumber \\
&\geq e^{-\beta \, \overline{C}^*_A} \, e^{\beta \, \mathbb{E}_A\left\{\overline{C}(\mathbf{x})\right\}} = e^{\beta \left( \mathbb{E}_A\left\{\overline{C}(\mathbf{x})\right\} - \overline{C}^*_A \right)}.
\end{align}
Thus, by taking natural logarithm,
\begin{equation}
\ln\left( \frac{|A|}{|A^*|} \right) \geq \beta \left( \mathbb{E}_A\left\{\overline{C}(\mathbf{x})\right\} - \overline{C}^*_A \right).
\end{equation}
Consequently,
\begin{equation}
\mathbb{E}_A\left\{\overline{C}(\mathbf{x})\right\} - \overline{C}^*_A \leq \frac{1}{\beta} \ln\left( \frac{|A|}{|A^*|} \right),
\end{equation}
which completes the proof.
\end{proof} 

This theorem indicates the convergence of the MC\textsuperscript{2} algorithm to a near-optimal solution. As the sample count increases, the mean value of the objective function over the sampled states should converge to its expected value since the recurrent class entered is ergodic and satisfies invariance \cite{mackay_it}. The minimum value obtained by the Markov chain should be less than this mean value. Hence, it should also be less than the expected value, as it falls within the range $\left[\overline{C}^*_A, \overline{C}^*_A + \frac{1}{\beta} \, \ln\left(\frac{|A|}{|A^*|} \right)\right]$. Observe that
\begin{equation} 
\frac{1}{\beta} \, \ln\left(\frac{|A|}{|A^*|} \right) \leq \frac{1}{\beta} \, \ln(z^{\gamma\kappa}) = \frac{\gamma\kappa}{\beta} \, \ln z.
\end{equation}
for the lifting problem.

\begin{remark}
Although increasing $ \beta $ appears to narrow this range and bring its maximum closer to the local minimum, it also makes transitions in the Markov chain more challenging and decreases the rate of convergence. This occurs because, with a higher $ \beta $, states with lower objective function values gain significantly higher weights in the transition probabilities, making the algorithm more prone to greediness and potentially leading it to suboptimal solutions. This could increase the number of iterations required to achieve a satisfactory result.
\end{remark}

At this stage, we define $f(\beta)$, the \textit{acceptance rate}, as the ratio of the number of transitions between distinct states to the total number of transitions. This metric measures the flow of state transitions, with higher $\beta$ making the Markov chain more greedy, increasing the likelihood of remaining in the same state, and indicating a lower acceptance rate. Conversely, lower $\beta$ reduces the effect of the objective function on state probabilities, making states nearly equiprobable and indicating a higher acceptance rate.

Within the algorithm, a target acceptance rate is determined and $\beta$ is adaptively adjusted to achieve the target acceptance rate. If the current acceptance rate exceeds the target, $\beta$ is increased; if it falls below the target, $\beta$ is decreased. This adjustment continues until the acceptance rate converges to the target rate. The initial $\beta$ has minimal impact due to rapid updates during early iterations.

We successively reduce the initial target acceptance rate after a predetermined number of iterations. This allows for broader state exploration in the early stages of optimization, followed by a more greedy search for optimal results in later stages. This approach implements a special form of \textit{simulated annealing}~\cite{kirkpatrick_sa}.

\subsection{Algorithm Description}

The pseudo-code for the general case of our MC\textsuperscript{2} algorithm is presented in Algorithm~\ref{algo:MC_2}. For simplicity, this algorithm assumes a constant target acceptance rate, but can be easily adapted for successive reduction in the target acceptance rate.

\begin{algorithm}
\caption{Markov Chain Monte Carlo (MC\textsuperscript{2}) Optimizer for Object Count Reduction} \label{algo:MC_2}
\begin{algorithmic}[1]
\Statex \textbf{Inputs:} $\mathcal{L}$: list of objects of interest; $d$: cardinality of index tuples; $\mathbf{x}_{\text{init}}$: initial input vector; $\beta_{\text{init}}$: initial value of $\beta$; $\mathcal{T}$: maximum transition count (maximum number of iterations); $\mathbf{a}$: set of possible values for each entry of the input vector;  $f_T$: target acceptance rate.
\Statex \textbf{Outputs:} $\overline{C}_{\text{opt}}$: minimum value of the objective function recorded during iterations; $\mathbf{x}_{\text{opt}}$: optimal value of the input vector resulting in the minimum objective function; $i$: number of transitions.
\Statex \textbf{Intermediate Variables:} $S$: list of index $d$-tuples; $\mathbf{x}$: input vector corresponding to the current state of the Markov chain; $i_a$: number of transitions to distinct states; $\beta$: hyper-parameter of the main distribution; $\nu$: index tuple of entries that are updated; $Z$: normalizing factor for probabilities; $\mathbf{x}', \mathbf{x}_{\text{prev}}$: input vectors used for intermediate calculations; $P_\nu ,P_\nu^*$: normalized and non-normalized mappings of conditional PMF values of the transition probability $P\left(\mathbf{x}' | \mathbf{x}\right)$ for the current indices $\nu$ such that $\mathbf{x}'$ and $\mathbf{x}$ differ only at entries indexed by $\nu$; $f$: acceptance~rate.
\State Initialize $S$ by going over $\mathcal{L}$ and calculating the common object counts for all entry pairs, then list $d$ indices of the most correlated entries for each entry. 
\State $\mathbf{x} \gets \mathbf{x}_{\text{init}}$.
\State Initialize $\overline{C}_{\text{opt}} \gets 1$, $\left(i,i_a\right) \gets \left(0,0\right)$, $\beta \gets \beta_{\text{init}}$.
\While{$i < \mathcal{T}$}
	\State Shuffle the order of $S$ randomly.
    \ForEach{$\nu \in S$}
        \State $i \gets i+1$, $Z \gets 0$, $\mathbf{x}_{\text{prev}}\gets \mathbf{x}$.
        \ForEach{$\mathbf{x}'_{\nu} \in \mathbf{a}^d$}
            \State Merge $\mathbf{x}'_{\nu}$  and $\mathbf{x}_{\backslash \nu}$ to obtain $\mathbf{x}'$. \textit{// Obtain entries at indices $\nu$ from $\mathbf{x}'_{\nu}$ and the rest from $\mathbf{x}_{\backslash \nu}$.}
            \If {$\mathbf{x}'$ does not satisfy the constraints}
            	\State $P_\nu(\mathbf{x}') \gets 0$.
            \Else
                \State Evaluate $\overline{C}(\mathbf{x}')$.
                \If{$\overline{C}(\mathbf{x}') < \overline{C}_{\text{opt}}$}
                    \State $\overline{C}_{\text{opt}}\gets \overline{C}(\mathbf{x}')$ and $\mathbf{x}_{\text{opt}}\gets \mathbf{x}'$.
                    \If{$\overline{C}_{\text{opt}} = 0$} \textbf{go to} Step~30.
                    \EndIf
                \EndIf
                \State $P_\nu^*(\mathbf{x}') \gets \exp(-\beta \overline{C}(\mathbf{x}'))$.
                \State $Z \gets Z + P_\nu^*(\mathbf{x}')$.
            \EndIf
        \EndFor
        \State $P_\nu \gets P_\nu^* / Z$. \textit{// This applies for all $\mathbf{x}$.}
        \State$\mathbf{x} \gets \mathrm{overrelaxSampling}(P_\nu)$.
        \If{$\mathbf{x} \neq \mathbf{x}_{\text{prev}}$} $i_a\gets i_a + 1$.
        \EndIf
    \EndFor
    \State $f = i_a/i$, $\beta\gets\mathrm{updateBeta}\left(\beta,f,f_T\right)$. \textit{// $\mathrm{updateBeta}$ and $\mathrm{overrelaxSampling}$ are functions we define.}
\EndWhile
\State \textbf{return} $\overline{C}_{\text{opt}}$, $\mathbf{x}_{\text{opt}}$, and $i$.
\end{algorithmic}
\end{algorithm}

In the partitioning case, the algorithm operates on the list of short cycle candidates in the base matrix and tries to minimize the weighted sum of short cycle candidate counts. We typically assign weights $\left(w_4, w_6, w_8\right) = \left(0, 1, 0.2\right)$, prioritizing cycle-$6$ candidates while discarding cycle-$4$ candidates since cycles-$4$ are easily addressed at the lifting stage. The initial partitioning is derived from the GD algorithm \cite{GRADE}, and the allowed state variables are constrained to be within $L_1$ and $L_\infty$ distances from the initial state to limit the search space. Neighboring states violating this constraint are assigned with zero transition probability at each transition.

In the lifting case, the algorithm minimizes short cycle or common substructure counts. It can operate on the list of cycle candidates in protograph corresponding to $\mathbf{H}^{\textup{g}}_{\textup{SC}}$, starting with a random or basic arrangement that eliminates all cycles-$4$. It then minimizes the count of active cycle-$6$ candidates, and subsequently active cycle-$8$ candidates if all cycle-$6$ candidates are inactive, ensuring the removal of all shorter cycles. Alternatively, the algorithm can focus on reducing the number of unlabelled elementary absorbing or trapping sets (UAS/UTS) \cite{ahh_md}, such as the $\left(4,4\left(\gamma-2\right)\right)$ UAS/UTS, while maintaining low-weight codewords and cycles-$4$ eliminated.\footnote{An $(a,d_1)$ UTS is a configuration with $a$ variable nodes (VNs), $d_1$ degree-$1$ check nodes (CNs), and no degree $> 2$ CNs. An $(a,d_1)$ UAS is a UTS such that each VN is adjacent to strictly more degree-$2$ than degree-$1$ CNs. In a binary code, these become elementary trapping/absorbing sets \cite{ahh_md}.}

Two self-defined functions are used in the pseudo-code of the algorithm, described as follows:
\begin{enumerate} [leftmargin=*]
    \item $\mathrm{updateBeta}\left(\beta, f, f_T\right)$ adjusts $\beta$ based on the current acceptance rate $f$ and the target acceptance rate $f_T$. A proportional integral derivative (PID) controller is employed to increase $\beta$ when $f > f_T$ and decrease it when $f < f_T$.
    \item $\mathrm{overrelaxSampling}(P_\nu)$ generates samples from the PMF associated with $P_\nu$ using \textit{ordered overrelaxation}, a method designed to mitigate the random walk behavior in Gibbs sampling~\cite{neal_mcmc}, with little complexity overhead for each transition. It is estimated to accelerate the convergence of Gibbs sampling by a factor of $10$ to $20$.

\end{enumerate}

\smallskip

The evaluation of the objective function involves iterating over the respective cycle candidate list and checking whether each cycle candidate is active or not. For partitioning, activeness is checked based on Equation~(\ref{eq:partition_condition}), while for lifting, activeness is determined using Equation~(\ref{eq:lifting_condition}). As an additional case for the lifting stage, to reduce $\left(4,4\left(\gamma-2\right)\right)$ UAS/UTS, the algorithm checks cycle-$8$ activeness and the existence of internal connections, counting only active cycle-$8$ candidates without any internal connections (see also \cite{channel_aware}). For each case, these counts are normalized by the maximum value.

\vspace{-0.1em}

\begin{remark}
The MC\textsuperscript{2} algorithm can also be adapted to target the more detrimental absorbing and trapping sets (ASs/TSs), which are key contributors to the error-floor phenomenon. These structures can be represented as disjunctive unions of fundamental cycles in a cycle basis \cite{ahh_md}, and they remain active only if all their fundamental cycles are active. Therefore, their patterns can be identified in the underlying graph along with their fundamental cycles, and their activeness can be assessed accordingly. The MC\textsuperscript{2} algorithm can then be applied in a similar manner to minimize the number of ASs/TSs.
\end{remark}

\vspace{-0.1em}

As mentioned earlier, evaluating conditional probabilities for a single transition requires $O(N^d)$ objective function evaluations, where $N$ is the number of possible values per entry and $d$ is the number of updated entries. Since each objective evaluation has a complexity of $O(|\mathcal{L}|)$, where $\mathcal{L}$ is the set of objects of interest, the overall complexity per transition is $O(N^d |\mathcal{L}|)$, and the total complexity is $O(\mathcal{T} N^d |\mathcal{L}|)$, with $\mathcal{T}$ denoting the maximum number of transitions.

To select $d$ and the target acceptance rate, we perform a grid search over a small set of values given the code parameters during early iterations and proceed with the best-performing choice, which causes only a minor computational overhead. In our experiments, $\mathcal{T}$ is set between $2{,}000\gamma\kappa$ and $20{,}000\gamma\kappa$, which are affordable in our method, to ensure computational feasibility. Finally, the algorithm terminates when an input vector achieves an objective function value of zero.

\vspace{-0.3em}
\section{Data Fitting and Estimation Analysis} \label{sec:fitting}

In this section, we discuss how data fitting can be applied to estimate the minimum value the Markov chain Monte Carlo (MC\textsuperscript{2}) algorithm can achieve and to provide a better understanding of its dynamics. 

\begin{definition} [Distribution of the objective function]
We introduce the probability distribution of the objective function under the main distribution of recurrent class $A$ as: 
\begin{align}
& P_{\overline{C}_A}(\mathcal{C}) = \mathbb{P}\left[\overline{C}(\mathbf{x}) = \mathcal{C} \mid \mathbf{x}\in A\right] = \sum_{\mathbf{x}\in A; \, \overline{C}(\mathbf{x}) =\mathcal{C}} P_A(\mathbf{x}) \nonumber \\
&= \sum_{\mathbf{x}\in A; \, \overline{C}(\mathbf{x}) =\mathcal{C}} \frac{e^{-\beta \overline{C}(\mathbf{x})}}{Z_A(\beta)} = \left|\left\{\mathbf{x} \mid \mathbf{x} \in A, \, \overline{C}(\mathbf{x}) = \mathcal{C} \right\}\right| \frac{e^{-\beta \mathcal{C}}}{Z_A(\beta)}.
\end{align}
\end{definition}

\vspace{-0.7em}
\begin{lemma}
The mean and standard deviation of ${\overline{C}_A}(\mathbf{x})$, which represents $\overline{C}(\mathbf{x})$ given $\mathbf{x} \in A$, are denoted by $\mu_A(\beta)$ and $\sigma_A(\beta)$, functions of $\beta$, and they satisfy 
\[
\frac{d}{d\beta} \mu_A(\beta) = -\sigma^2_A(\beta).
\]
\end{lemma}

\begin{proof}
The mean of ${\overline{C}_A}(\mathbf{x})$ can be written as 
\begin{equation} \label{eq:mu_A}
\mu_A(\beta) = \sum_{\mathbf{x}\in A} \overline{C}(\mathbf{x}) P_A(\mathbf{x}) =\sum_{\mathbf{x}\in A} \overline{C}(\mathbf{x})\frac{e^{-\beta \, \overline{C}(\mathbf{x})}}{Z_A(\beta)}.
\end{equation}
Using Equation~(\ref{eq:Z_A}), it can be shown that
\begin{equation}\label{eq:dZ_A}
\frac{d}{d\beta} Z_A(\beta)=-\sum_{\mathbf{x}\in A} \overline{C}(\mathbf{x}) \, e^{-\beta \overline{C}(\mathbf{x})} = - Z_A(\beta) \, \mu_A(\beta).
\end{equation}
Lastly, Equations~(\ref{eq:mu_A}) and~(\ref{eq:dZ_A}) can be combined to show that
\begin{align} \label{eq:mu_var}
&\frac{d}{d\beta}\mu_A(\beta) = \frac{d}{d\beta} \left(\sum_{\mathbf{x} \in A} \overline{C}(\mathbf{x}) \frac{e^{-\beta \, \overline{C}(\mathbf{x})}}{Z_A(\beta)}\right) \nonumber \\
&= \sum_{\mathbf{x} \in A} \overline{C}(\mathbf{x}) \frac{d}{d\beta} \left(\frac{e^{-\beta \, \overline{C}(\mathbf{x})}}{Z_A(\beta)} \right) \nonumber \\
&= \sum_{\mathbf{x} \in A} \left( -\overline{C}^2(\mathbf{x}) \frac{e^{-\beta \, \overline{C}(\mathbf{x})}}{Z_A(\beta)} + \overline{C}(\mathbf{x}) \, \mu_A(\beta) \frac{e^{-\beta \, \overline{C}(\mathbf{x})}}{Z_A(\beta)} \right) \nonumber \\
&= - \left(\sum_{\mathbf{x} \in A} \overline{C}^2(\mathbf{x}) P_A(\mathbf{x}) \right) + \mu_A(\beta) \left(\sum_{\mathbf{x} \in A} \overline{C}(\mathbf{x}) P_A(\mathbf{x}) \right) \nonumber \\
&= - \left[ \mathbb{E}_A\left[\overline{C}^2_A(\mathbf{x})\right] - \left(\mathbb{E}_A\left[\overline{C}_A(\mathbf{x})\right]\right)^2\right] = -\sigma_A^2(\beta).
\end{align}
\end{proof} 

To analyze the effect of $\beta$, we adapt the MC\textsuperscript{2} optimization algorithm to generate a sequence of $\overline{C}(\mathbf{x})$ values for statistical analysis. The algorithm is executed for predetermined and fixed values of $\beta$. For each fixed $\beta$, we collect sequences, generate histograms of the collected data, and record the acceptance rate. We then fit the data to known probability distributions based on the generated histograms in order to evaluate the effect of $\beta$.

Our results show that for high object populations, the distribution of the objective function closely follows a discrete Gaussian distribution, reasonably approximated by a continuous Gaussian. We further observe that the acceptance rate $f(\beta)$ can be approximated by a decreasing degree-$2$ polynomial for $\beta \geq 0$ in the region of interest. Above a threshold for $\beta$, the Markov chain gets trapped in a local optimum, disrupting the fitting distribution. In this regime, $f(\beta)$ oscillates near zero, approximated in our remaining derivations as $f(\beta) \approx 0$.

It should be noted that the approximated continuous Gaussian distribution of the objective function should be clipped at $\overline{C}_A^*$, since the objective function cannot take smaller values. For sufficiently small values of $\beta$, the probability of the objective function taking values below its minimum in an unclipped continuous Gaussian is infinitesimally small and can be neglected. Therefore, we use the unclipped Gaussian distribution to fit the collected data and select $\beta$ values from a region where they are small enough to make the difference between the clipped and unclipped distributions negligible.

\begin{remark}
It has been shown in \cite{reins_gdmd} and \cite{banihashemi2} that for short cycles, the probability of remaining active after random partitioning (resp., lifting) is of order $1/m$ (resp., $1/z$), under the assumption that the partitioning (resp., lifting) matrix entries are drawn from a uniform, independent, and identical distribution (i.i.d.). However, in the MC\textsuperscript{2} framework, this assumption does not strictly hold, as the probability distribution depends on the number of active objects, which share non-zero entries, and therefore introduce dependencies. 

At $\beta = 0$, the entry distribution is approximately uniform, since all feasible arrangements are equally likely, although the problem constraints still impose dependencies on the distribution. Let $S_i$ denote a Bernoulli random variable representing the activeness of the $i$-th object in the object list, and let $C(\mathbf{x}) = \sum_i S_i$ denote the unnormalized objective function. For small $\beta$, it is reasonable to approximate the mean of $S_i$ as $O(1/m)$ or $O(1/z)$, while acknowledging that the variables $S_i$ are dependent due to shared entries among cycles. By the central limit theorem, the distribution of $C(\mathbf{x})$ can then be approximated as a Gaussian distribution, which aligns well with our empirical observations, particularly for large object sets. See Fig.~\ref{fig3} for the resemblance between the distribution of the numerical data generated by the MC\textsuperscript{2} framework and the Gaussian distribution.
\end{remark}

We consider the relationship between the mean and the variance of the fitted distributions. We realize that collected mean and variance pairs, as they vary with $\beta$, could be fitted to a linear relationship. This observation leads us to the equation:
\begin{equation} \label{eq:mu_var2}
    \sigma_A^2(\beta) \approx k \mu_A(\beta) + l,
\end{equation}
$\exists k > 0$ and $\exists l < 0$.

Combining Equations~(\ref{eq:mu_var}) and~(\ref{eq:mu_var2}), we deduce the following differential equation for the mean:
\begin{equation}
    \frac{d}{d\beta}\mu_A(\beta) = -\sigma_A^2(\beta) \approx -k \mu_A(\beta) - l.
\end{equation}
The solution to this differential equation is given by:
\begin{equation}
\mu_A(\beta) \approx \left(\mu_0 + \frac{l}{k}\right)e^{-k\beta} - \frac{l}{k},
\end{equation}
where $\mu_0 = \mu_A(0)$.

\begin{remark}
An independent fit of the mean as a function of $\beta$, named as $\overline{\mu}_A(\beta)$, also closely conforms to the exponential decay model:
\begin{equation}\label{eq:mu_est}
    \overline{\mu}_A(\beta) = c + a e^{-b\beta},
\end{equation}
for $a,b,c > 0$ obtained by fitting, further validating the model as an accurate representation of the data. Based on this, we will subsequently use $\overline{\mu}_A(\beta)$ as an estimate of the mean of $\overline{C}_A(\mathbf{x})$ and naturally assume that it satisfies the properties of the actual mean.

\end{remark}

Furthermore, we derive the expression for the estimated standard deviation of $\overline{C}_A(\mathbf{x})$ using Equations~(\ref{eq:mu_var}) and (\ref{eq:mu_est}):
\begin{equation}\label{eq:sigma_fit}
    \overline{\sigma}_A(\beta) = \sqrt{-\frac{d}{d\beta}\overline{\mu}_A(\beta)} = \sqrt{ab} \, e^{-\frac{b}{2}\beta}.
\end{equation}

\begin{remark} \label{remark_beta_vs_mu_var}
It can be observed that increasing $\beta$ gives more weight to lower $\overline{C}_A(\mathbf{x})$ values and their corresponding $\mathbf{x}$ vectors, for both the distribution of the objective function and the main distribution of the recurrent class $A$. Conversely, reducing $\beta$ diminishes the effect of the objective function, making all $\mathbf{x}$ vectors close to equally likely. If we consider the two corner cases: 
\begin{enumerate}
\item[1)] When $ \beta = 0 $, all $ \mathbf{x} $ vectors become equally likely, and the main distribution of the recurrent class $A$ becomes uniform over all $ \mathbf{x} \in A$.
\item[2)] For the other extreme, as $ \beta \to \infty $, the probabilities of the optimal values dominate, making $P_A(\mathbf{x})$ uniform over $A^*$ and zero elsewhere. Similarly, the distribution of the objective function converges to $ \delta\left[\mathcal{C} - \overline{C}_A^*\right] $.
\end{enumerate}

Therefore, it can be expected that the mean of $\overline{C}_A(\mathbf{x})$ decreases as $\beta$ increases, whereas the standard deviation is maximized when $\beta = 0$ and approaches zero as $\beta$ goes to infinity. These observations align with the form of our estimated functions for these parameters, supporting our fitting~results. Also, as $ \beta \to \infty$, $ P_{\overline{C}_A}(\mathcal{C}) $ converges to $ \delta\left[\mathcal{C} - \overline{C}^*_A\right] $ and
\begin{equation}
\overline{C}^*_A = \lim_{\beta \to \infty} \mu_A(\beta) \approx \lim_{\beta \to \infty} \overline{\mu}_A(\beta) = c,
\end{equation}
showing that parameter $c$ in the fitting function could be used as an estimate for the minimum of the objective function. We share the results of this estimate in Table~\ref{table_3}, Section~\ref{sec:numeric}.
\end{remark}

Additionally, we can use the outcomes of the fitting results to offer insights about the expected iteration count of the algorithm to reach a certain sub-minimum value within a given distance from the optimum, as well as the cardinality of the recurrent classes that the algorithm enters. First, we assume that the objective function distribution on the recurrent class $A$  for different $\beta$ values can be approximated as a continuous Gaussian distribution with a probability density function (PDF) given by 
\begin{equation}\label{eq:pdf_fit}
f_{\overline{C}_A}(\mathcal{C}) = \frac{1}{\sqrt{2\pi \overline{\sigma}_A^2(\beta)}} \; \exp\left(-\frac{(\mathcal{C}-\overline{\mu}_A(\beta))^2}{2\overline{\sigma}_A^2(\beta)}\right).
\end{equation}

We assume that the fitting parameter corresponding to the estimation of the minimum value, $c$, is close enough to the actual minimum such that we can write
\begin{equation} \label{eq:mu_fit}
\overline{\mu}_A(\beta) = \overline{C}_A^* + a \, e^{-b\beta}.
\end{equation}
Then, we can combine Equations~(\ref{eq:sigma_fit}),~(\ref{eq:pdf_fit}),  and~(\ref{eq:mu_fit}) to approximate the probability that the objective function is within a certain bounded distance, $\epsilon$, from the optimum as
\begin{align}\label{eq:p_c_epsilon1}
&\mathbb{P}\left[\overline{C}_A(\mathbf{x})-\overline{C}^*_A \leq \epsilon \right] \approx \int_{\overline{C}^*_A}^{\overline{C}^*_A+\epsilon} f_{\overline{C}_A}(\mathcal{C}) d\mathcal{C}  \nonumber\\
&= Q\left(\frac{\overline{C}^*_A - \overline{\mu}_A(\beta)}{\overline{\sigma}_A(\beta)} \right) - Q\left(\frac{\overline{C}^*_A + \epsilon - \overline{\mu}_A(\beta)}{\overline{\sigma}_A(\beta)} \right)\nonumber \\
&= Q\left(- \sqrt{\frac{a}{b}} \, e^{-\frac{b}{2}\beta}\right) - Q\left(- \sqrt{\frac{a}{b}} \, e^{-\frac{b}{2}\beta} + \frac{\epsilon}{\sqrt{ab}} \, e^{\frac{b}{2}\beta} \right),
\end{align}
which we can evaluate after fitting, for any value of $\beta$ and $\epsilon$. 
Here $Q(x)$ is defined as:
\begin{align}
Q(x) = \frac{1}{\sqrt{2\pi}}\int_x^\infty \exp\left(-\frac{u^2}{2}\right)du.
\end{align}
Furthermore, Equation~(\ref{eq:p_c_epsilon1}) can be effectively approximated for most of the ranges of $\beta$ and $\epsilon$ by replacing the $Q$ function with the expression in Equation~(\ref{eq:q_approx}) below \cite{leon_garcia}

\begin{align}\label{eq:q_approx}
Q(x) = \frac{1}{\sqrt{2\pi}} \left[\frac{1}{\left(1-\frac{1}{\pi}\right)x + \frac{1}{\pi} \sqrt{x^2 + 2\pi}}\right] e^{-\frac{x^2}{2}}.
\end{align}

However, both Equations~(\ref{eq:p_c_epsilon1}) and~(\ref{eq:q_approx}) are susceptible to floating-point precision errors for small values of $\beta$ and $\epsilon$. These operational ranges of $\beta$ and $\epsilon$ become instrumental for certain estimation processes, which will be discussed later in this section. In such cases, a first order simplification using the Taylor series expansion of the second $Q$ function in Equation~(\ref{eq:p_c_epsilon1}) has the form 
\begin{align} \label{eq:p_c_epsilon2}
\mathbb{P}\left[\overline{C}_A(\mathbf{x})-\overline{C}^*_A \leq \epsilon\right] \approx \frac{\epsilon}{\sqrt{2\pi a b}} \, \exp\left(\frac{b}{2}\beta - \frac{a}{2 b} \, e^{-b \beta} \right),
\end{align} 
and it is more resilient to floating-point errors and could serve as an effective alternative for small $\beta$ and $\epsilon$ ranges.

This approximated probability can be used to estimate the order of the iteration count of the MC\textsuperscript{2} algorithm with fixed $\beta$ to reach a suboptimal value within difference $\epsilon$ from the optimum of $A$. This count is denoted by $\mathrm{iter}\left(\beta,\epsilon,A\right)$. We also introduce the set $A_{\epsilon} = \left\{ \mathbf{x} \mid \overline{C}(\mathbf{x})-\overline{C}^*_A \leq \epsilon , \mathbf{x}\in A\right\} $, which is the set that we are trying to sample from. For ideal sampling, i.e., if we assume a geometric distribution argument for the iteration count, $\mathbb{E}\left[\mathrm{iter}\left(\beta,\epsilon,A\right)\right]$ should be the inverse of $ \mathbb{P}\left[\overline{C}_A(\mathbf{x})-\overline{C}^*_A \leq \epsilon \right] $. However, our Markov chain does not behave as fluidly as an ideal sampler in terms of its transitions between consecutive states. When the transition rate is low, this can negatively impact the required number of iterations. To account for this, we incorporate an estimate of the acceptance rate, $\overline{f}(\beta)$, into the iteration count order estimate and inversely proportional to it such that this factor is reflected. Thus, the iteration count estimate is
\begin{equation}
\mathbb{E}\left[\mathrm{iter}\left(\beta, \epsilon, A\right)\right] \sim O\left( \frac{1}{\overline{f}(\beta) \, \mathbb{P}\left[\overline{C}_A(\mathbf{x})-\overline{C}^*_A \leq \epsilon \right]} \right).
\end{equation}
Observe that the iteration expectation is not directly equal to the number in the order argument, but it is expected that the iteration count is in this order.

Lastly, we can use Equation~(\ref{eq:p_c_epsilon1}) to estimate the order of the cardinality of the recurrent class that the algorithm enters. We focus on the $ \beta = 0 $ case. It follows from Equation~(\ref{eq:Z_A}) that $ Z_A(0) = |A| $. Additionally, $ P_A(\mathbf{x}) $ becomes uniform over $ A $ in this case. Therefore,
\begin{equation}
\mathbb{P}\left[\overline{C}_A(\mathbf{x})-\overline{C}^*_A \leq \epsilon \right] =\frac{|A_{\epsilon}|}{|A|}.
\end{equation}
Additionally, the expression in Equation~(\ref{eq:p_c_epsilon1}) becomes
\begin{equation}
\mathbb{P}\left[\overline{C}_A(\mathbf{x})-\overline{C}^*_A \leq \epsilon \right] \approx Q\left(- \sqrt{\frac{a}{b}}\right) - Q\left(- \sqrt{\frac{a}{b}}  + \frac{\epsilon}{\sqrt{a b}}\right) 
\end{equation}
at $\beta=0$. These two equations can be combined to reach
\begin{equation}
|A| = \frac{|A_{\epsilon}|}{Q\left(- \sqrt{\frac{a}{b}}\right) - Q\left(- \sqrt{\frac{a}{b}}  + \frac{\epsilon}{\sqrt{a b}}\right)}.
\end{equation}
For a small $\epsilon$, we assume that $|A_{\epsilon}| \sim O\left(1\right)$, hence:
\begin{equation}\label{eq:cardinality}
|A| \sim O\left( \frac{1}{ Q\left(- \sqrt{\frac{a}{b}}\right) - Q\left(- \sqrt{\frac{a}{b}}  + \frac{\epsilon}{\sqrt{ab}}\right)} \right).
\end{equation}

Given that $C(\mathbf{x})$ takes integer values in the lifting stage, $\overline{C}(\mathbf{x})$ takes integer multiples of $1/\alpha$. Consequently, $1/\alpha$ is a suitable bin size for discretizing $f_{\overline{C}_A}(\mathcal{C})$ to approximate $P_{\overline{C}_A}(\mathcal{C})$. Because of that, we set $\epsilon = 1/\alpha$ for the estimations at the lifting stage. Combining this with Equations~(\ref{eq:p_c_epsilon2}), and (\ref{eq:cardinality}), we obtain:
\begin{equation}\label{eq:cardinality2}
|A| \sim O\left(\alpha\sqrt{2\pi ab} \, \exp\left(\frac{a}{2b}\right)\right).
\end{equation}
We can conduct a similar analysis as well for the partitioning case, which we skip for brevity.

Observe that variations in $\beta$ alter the value of transition probabilities in the Markov chain but do not affect the connectivity of states, and hence the recurrent classes, provided that $\beta$ does not go to infinity. Therefore, Equation~(\ref{eq:cardinality2}) is a valid estimate of the cardinality order for all finite $\beta$ values.

Finally, we observed that reducing the sample size to a certain extent does not compromise parameter estimation or model fitting. We utilize $5$ distinct $\beta$ values, generating $100 \gamma \kappa$ samples per $\beta$ and resulting in a total of $500 \gamma \kappa$ samples for estimation. The results of these estimations are presented in more detail in Section~\ref{sec:numeric}.

\section{Experimental Results} \label{sec:numeric}

In this section, we present the performance evaluation of our Markov chain Monte Carlo (MC\textsuperscript{2}) optimizer, demonstrating the gains it offers across various codes compared with the literature in terms of reduced object counts and error/erasure-rate performance, and confirming the accuracy of the derived estimate based on data fitting. All codes here are binary codes.

We used optimal overlap (OO) partitioning \cite{oocpo} for low memory codes ($m\leq 2$) and for topologically-coupled (TC) codes \cite{GRADE}. Otherwise, gradient-descent supported MC\textsuperscript{2} (GD-MC\textsuperscript{2}) algorithm is used for partitioning. For all of the codes, circulant power optimization is done by the MC\textsuperscript{2} algorithm. 

\begin{table}
\caption{Parameters, Lengths, and Rates of Proposed SC/TC Codes}
\vspace{-0.5em}
\centering
\scalebox{1.00}
{
\begin{tabular}{|c|c|c|c|c|c|c|c|}
\hline
\makecell{Code name} & \makecell{$\gamma$} & \makecell{$\kappa$} & \makecell{$z$} & \makecell{$L$} & \makecell{$m$} & \makecell{Length } & \makecell{Rate} \\
\hline
SC Code~1 & $3$ & $17$ & $17$ & $30$ & $1$ & $8{,}670$ & $0.82$ \\
\hline
SC Code~2 & $3$ & $17$ & $17$ & $30$ & $2$ & $8{,}670$ & $0.81$ \\
\hline
SC Code~3 & $3$ & $17$ & $7$ & $10$ & $9$ & $1{,}190$ & $0.66$ \\
\hline
SC Code~4 & $4$ & $29$ & $29$ & $20$ & $19$ & $16{,}820$ & $0.73$ \\
\hline
SC Code~5 & $3$ & $7$ & $11$ & $30$ & $5$ & $2{,}310$ & $0.50$ \\
\hline
SC Code~6 & $4$ & $17$ & $37$ & $10$ & $1$ & $6{,}290$ & $0.74$ \\
\hline
TC Code~1 & $4$ & $17$ & $17$ & $50$ & $4$ & $ 14{,}450$ & $0.75$ \\
\hline
\end{tabular}}
\label{table_1}
\vspace{-0.5em}
\end{table}

In Table~\ref{table_1}, we list parameters, lengths, and design rates~for six spatially-coupled (SC) codes and one TC code. Our code list covers a wide range of lengths and design rates, demonstrating that the proposed method is effective for designing codes suitable for diverse applications, including communication systems. Construction methods and counts of objects of interest for these codes are provided in Table~\ref{table_2}. All codes are free of cycles-$4$, and those listed under cycle-$8$ counts are also free of cycles-$6$. 
We also refer the reader to literature papers from which we obtain some of our codes' partitioning matrix constructed by the OO method. We compare our codes with literature counterparts that share the same parameters, which are listed in Table~\ref{table_1} for all codes (thick lines in relevant tables separate pairs of codes we compare). We have modified the $L$ parameter of the literature code \cite[$(3,17)$]{GRADE}, while keeping the other code parameters and the partitioning/lifting matrices unchanged, in order to match its parameters with those of SC~Code~3. The rationale for using a lower $L$ is to obtain shorter codes suitable for wireless communication applications, thereby demonstrating the versatility of our proposed MC\textsuperscript{2} method. Note that the optimization process during the partitioning stage is not affected by the number of replicas, since our goal is to minimize the number of active cycle candidates, which is independent of $L$. We also conducted the lifting-stage optimization separately for high and low $L$ values, and observed that the resulting cycle counts differ by less than $1\%$, indicating that the choice of $L$ only has a negligible impact on this stage.

\begin{remark}
Our MC\textsuperscript{2} method can also be adapted for finite-length (FL) optimization of multi-dimensional (MD) codes via managing relocations to reduce short cycle counts. It can be integrated with probabilistic frameworks~\cite{reins_gdmd} similarly to the case of partitioning. To demonstrate this, we included MD-SC~Code~1 with parameters $(\gamma, \kappa, z, L, m) = (4, 17, 17, 10, 1)$ and $3$ auxiliary matrices, resulting in a code of length $8{,}670$ and rate $0.74$. Partitioning is done using the OO method; lifting and relocations are handled by the MC\textsuperscript{2} method, targeting cycles-$6$. MD-SC~Code~1 has $3{,}366$ cycles-$6$, compared with $6{,}375$ and $9{,}078$ in two literature codes \cite{reins_gdmd}, \cite{homa-lev} with the same parameters and partitioning matrix, showing the effectiveness of our approach.
\end{remark}

\begin{table}
\caption{Partitioning/Lifting/Relocation Methods, and Reduced Object Counts of Proposed SC/TC Codes (Code Parameters Are Listed in Table~\ref{table_1})}
\vspace{-0.5em}
\centering
\scalebox{1.00}
{
\begin{tabular}{|c|c|c|c|}
\hline
\makecell{Code name} & \makecell{Partitioning} & \makecell{Lifting} & \makecell{Cycle-$6$ count} \\
\hline
\cite[SC Code~3]{oocpo} & OO \cite{oocpo} & \cite{oocpo} & $14{,}960$ \\
\hline
SC Code~1 & OO \cite{oocpo} & MC\textsuperscript{2} & $12{,}937$ \\
\thickhline
\cite[SC Code~7]{oocpo} & OO \cite{oocpo} & \cite{oocpo} & $0$ \\
\hline
SC Code~2 & OO \cite{oocpo} & MC\textsuperscript{2} & $0$ \\
\thickhline
\cite[$\left(4,17\right)$]{GRADE} & TC-OO \cite{GRADE} & \cite{GRADE} &  $15{,}436$ \\
\hline
TC Code~1 & TC-OO \cite{GRADE} & MC\textsuperscript{2} & $8{,}007$ \\
\hline
\hline
\makecell{Code name} & \makecell{Partitioning} & \makecell{Lifting} & \makecell{Cycle-$8$ count} \\
\hline
\cite[$\left(3,17\right)$]{GRADE} & GRADE-AO \cite{GRADE} & \cite{GRADE} & $13{,}860$ \\
\hline
SC Code~3 & GD-MC\textsuperscript{2} & MC\textsuperscript{2} & $12{,}530$ \\
\thickhline
\cite[$\left(4,29\right)$]{GRADE} & GRADE-AO \cite{GRADE} & \cite{GRADE} & $528{,}090$ \\
\hline
SC Code~4 & GD-MC\textsuperscript{2} & MC\textsuperscript{2} & $409{,}973$ \\
\thickhline
SC Code~5 & GD-MC\textsuperscript{2} & MC\textsuperscript{2} & $0$ \\
\hline
\hline
\makecell{Code name} & \makecell{Partitioning} & \makecell{Lifting} & \makecell{$\left(4,8\right) $ UTS count} \\
\hline
\cite[Code 10]{channel_aware} & OO \cite{channel_aware} & \cite{channel_aware} & $1{,}253{,}177$ \\
\hline
SC Code~6 & OO \cite{channel_aware} & MC\textsuperscript{2} & $1{,}229{,}177$ \\
\hline
\end{tabular}}
\label{table_2}
\vspace{-0.5em}
\end{table}

We also present the actual and estimated minimum object counts, as well as the recurrent class cardinality order estimations (on base-$z$ $\log$ scale) derived from data fitting at the lifting stage in Table~\ref{table_3}. We exclude SC Code~2 and SC Code~5 from such estimations since the MC\textsuperscript{2} algorithm is able to remove all the objects of interest.

We compare the computational complexity of our method with literature methods for both partitioning and lifting. Table~\ref{table_4} presents the number of evaluations required for each input to compute objective function values and check constraint satisfaction. We focus on the evaluations needed to achieve results close to those we report. For our method, the evaluation counts are represented in split cells: the left value indicates evaluations required to reach our optimal result, while the right value shows evaluations required to match literature results. For SC Code~2 and SC Code~5, we only share one count since we only consider the evaluation count to remove all objects of interest. Lastly, evaluation counts of partitioning are only reported for codes constructed using GD-based methods, since the OO method gives the final FL partitioning, and thus, no algorithmic optimization is required for this stage.

\begin{figure}
    \centering
    \includegraphics[width=0.5\textwidth]{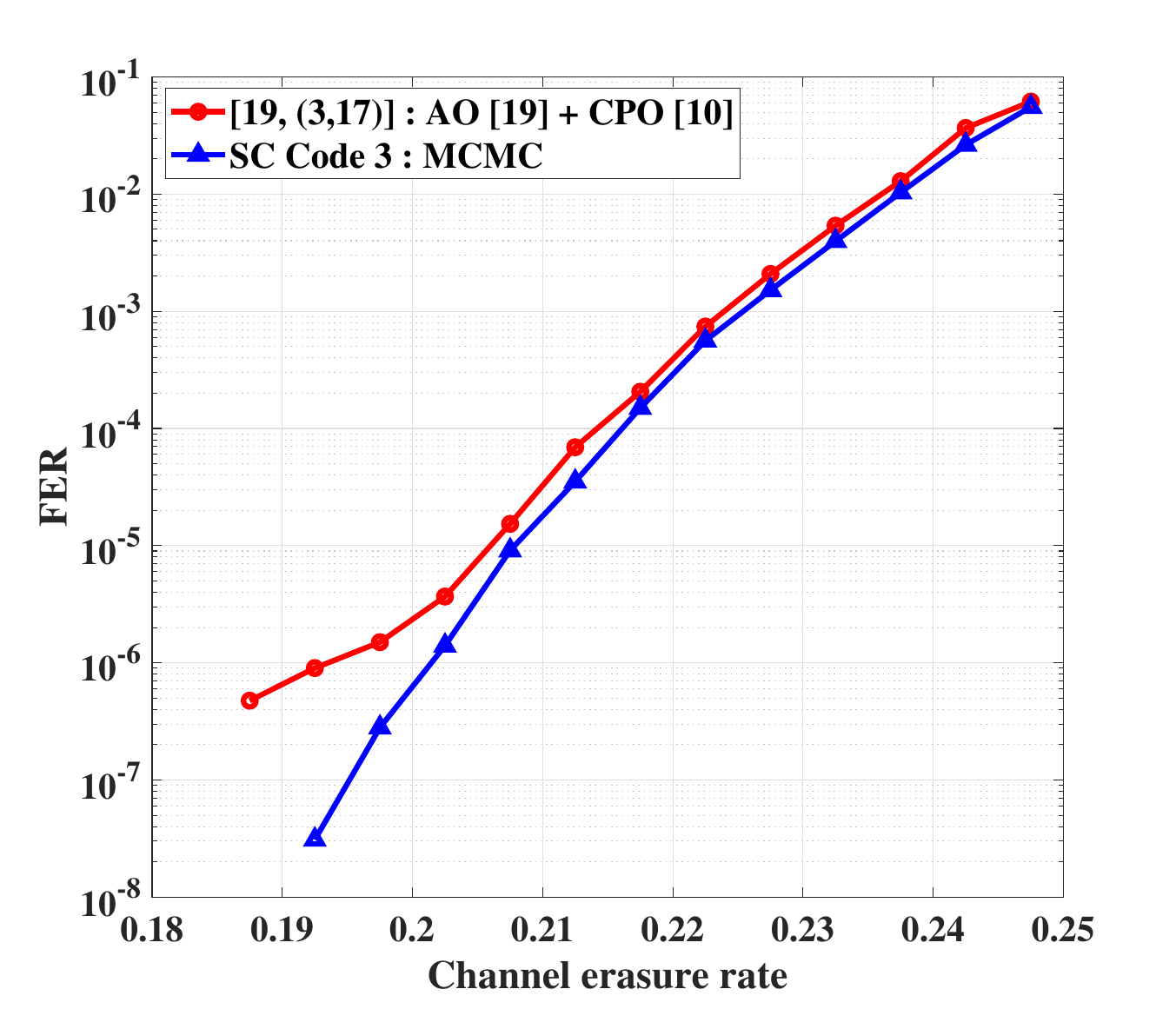}
    \caption{FER curve of SC Code~3, along with its literature counterpart \cite[$(3,17)$]{GRADE}, over the BEC. Both curves correspond to codes with parameters $\left(\gamma, \kappa, z, L, m\right) = \left(3,17,7,10,9\right)$.}
    \label{fig1}
\vspace{-1.0em}
\end{figure}

We also compare the computational latency of two of our code constructions with existing literature, focusing specifically on the run-times of algorithms. Relevant results are presented in Table~\ref{table_5}, where we list termination times measured in minutes. For each case, we provide both the time taken to reach our optimal solution in the left cell, and time to match the result from the literature in the right cell. For reproducibility, all software implementations were executed in MATLAB on a Lenovo ThinkSystem SR650. This system is equipped with an Intel Xeon Gold 6226R CPU at $2.90$ GHz with $64$ cores.

SC Code~3 and SC Code~4, along with their literature counterparts, are simulated over the binary erasure channel (BEC) and the additive white Gaussian noise channel (AWGNC) using a peeling decoder and a sum-product (fast Fourier transform based) decoder \cite{GRADE}, respectively. The results are presented in Fig.~\ref{fig1} and Fig.~\ref{fig2}.

Lastly, we carried out the data acquisition and fitting process described in Section~\ref{sec:fitting} for SC Code~4 at the lifting stage. In particular, we collected cycle-$8$ count data for three different $\beta$ values ($0$, $3977$, and $9782$), and fitted the resulting histograms of the collected data to Gaussian distributions. The fitted curves and the corresponding histograms are shown in Fig.~\ref{fig3}. For comparison, we also included the case where a set of lifting matrices were generated randomly, with each matrix entry drawn from a uniform, independent, and identical distribution (i.i.d.), and the resulting normalized cycle-$8$ counts were collected. It can be observed from these plots that the distribution of the collected data closely resembles a Gaussian distribution, as mentioned earlier, in all the cases. The plots illustrate the change in the mean and variance of the distribution with increasing $\beta$, as discussed in Remark~\ref{remark_beta_vs_mu_var}. The plots also highlight the differences between the distributions produced by the MC\textsuperscript{2} framework and the uniform i.i.d.\ case, where each distribution exhibits distinct mean and variance characteristics.

\begin{figure}
    \centering
    \includegraphics[width=0.48 \textwidth]{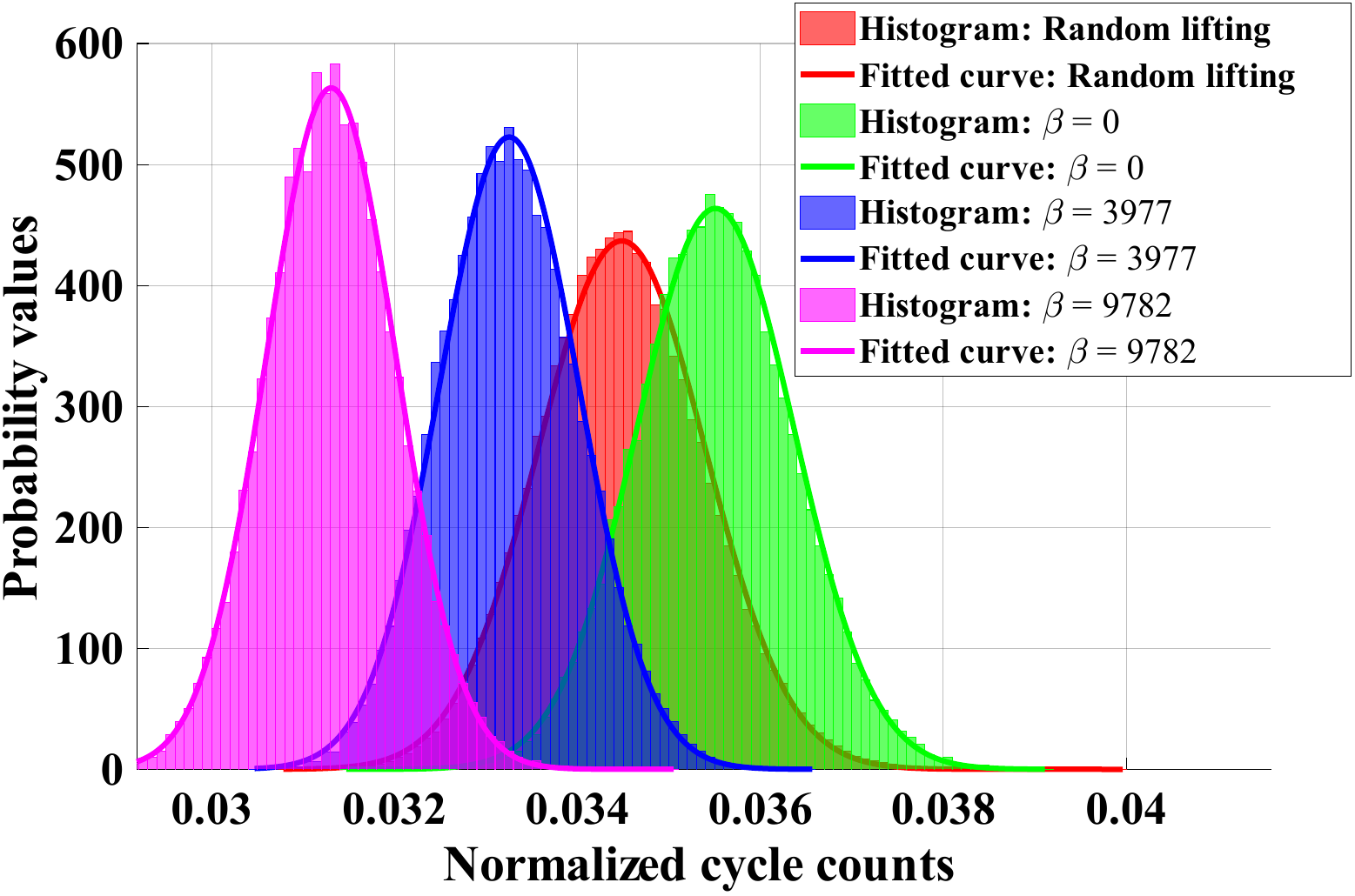}
    \caption{Data acquisition and fitting results for SC Code~4 at the lifting stage. The code parameters are $\left(\gamma, \kappa, z, L, m\right) = \left(4, 29, 29, 20, 19\right)$. Three data sets were collected for different $\beta$ values ($0$, $3977$, and $9782$), and they are shown in green, blue, and magenta, respectively. For comparison, a fourth curve is included, which is obtained by collecting cycle counts resulting from lifting matrices whose entries are drawn from a uniform i.i.d.\ (shown in red).}
    \label{fig3}
\vspace{-1.2em}
\end{figure}

\begin{table}[b]
\caption{Estimated and Actual Minimum Results With Recurrent Class Cardinality Estimations (Code Parameters Are Listed in Table~\ref{table_1})}
\vspace{-0.5em}
\centering
\scalebox{1.00}{
\begin{tabular}{|c|c|c|c|}
\hline
\makecell{Code name} & \makecell{Minimum \\ object count} & \makecell{Estimated \\ object count} & \makecell{Recurrent class \\ cardinality estimation} \\
\hline
SC Code~1 & $12{,}937$ & $13{,}399$ & $14.70$ \\
\hline
SC Code~3 & $12{,}530$ & $12{,}096$ &  $13.31$ \\
\hline
SC Code~4 & $409{,}973$ & $390{,}953$ & $29.51$ \\
\hline
SC Code~6 & $1{,}229{,}177$ & $1{,}175{,}229$ & $30.73$\\
\hline
TC Code~1 & $8{,}007$ & $11{,}757$ & $17.63$\\
\hline
\end{tabular}}
\label{table_3}
\end{table}

\begin{figure}
   \vspace{-1.0em}
    \centering
    \includegraphics[width=0.5\textwidth]{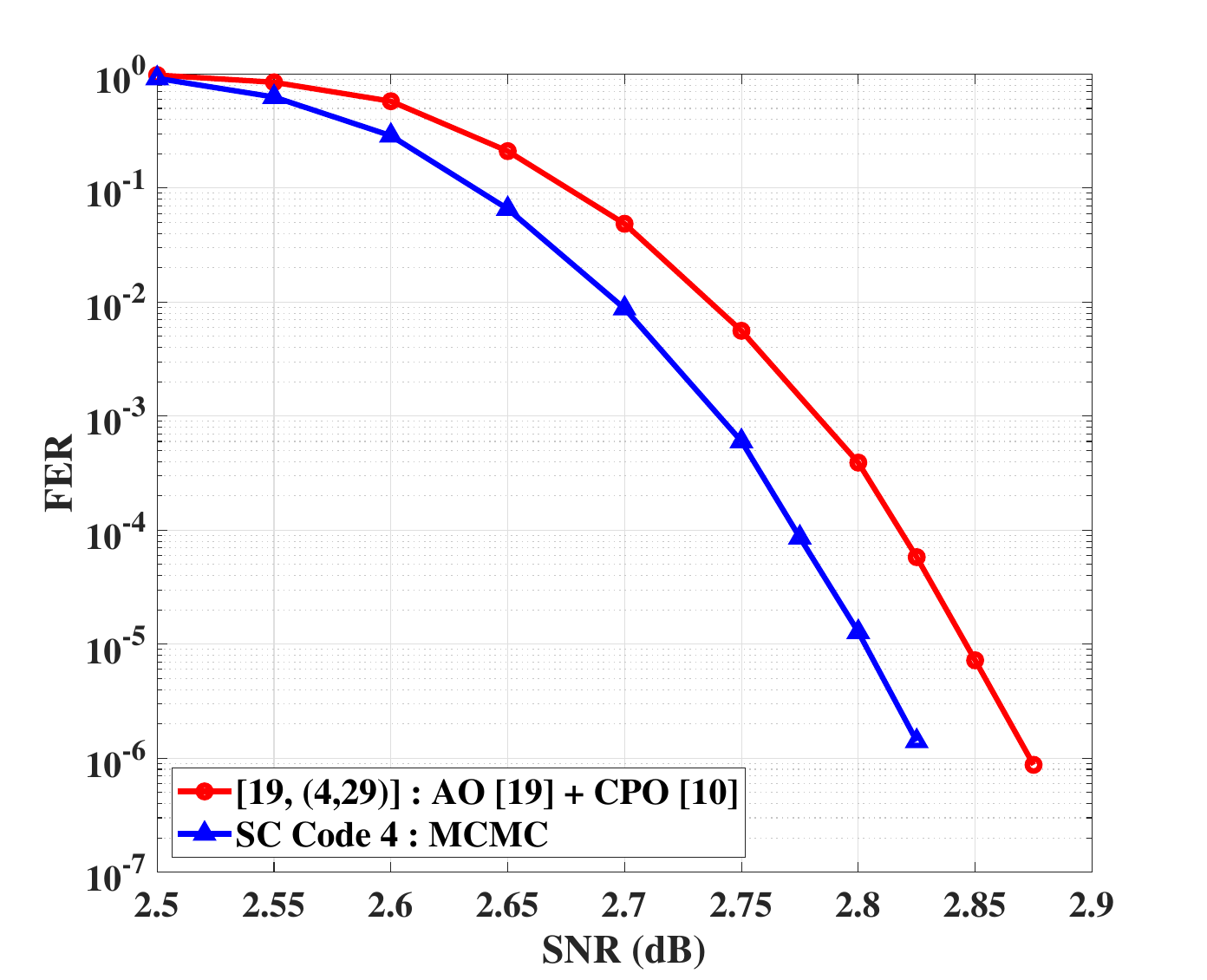}
    \caption{FER curve of SC Code~4, along with its literature counterpart \cite[$(4,29)$]{GRADE}, over the AWGNC. Both curves correspond to codes with parameters $\left(\gamma, \kappa, z, L, m\right) = \left(4,29,29,20,19\right)$.}
    \label{fig2}
\vspace{-1.0em}
\end{figure}

Here are some discussions summarizing the main conclusions from our numerical results:

\begin{enumerate}
    \item \textbf{Reduction of short cycle and common substructure counts:} 
The MC\textsuperscript{2} algorithm effectively reduces the number of objects of interest, including cycles-$6$, cycles-$8$, and $\left(4,8\right)$ UTSs, outperforming literature methods:
\begin{itemize}
    \item For cycles-$6$, count reductions range from $13.52\%$ to $48.13\%$ compared with respective literature counts. For instance, MD-SC Code~1 has a reduced count of $3{,}366$ compared with $6{,}375$ ($47.20\%$), demonstrating the gains of the algorithm when applied for lifting and relocations, while TC Code~1 achieves a $48.13\%$ reduction.
    \item Cycle-$8$ count reductions range from $9.60\%$ to $22.37\%$. Notably, SC Code~3 has $9.60\%$ less objects than its counterpart, and SC Code~4 has $118{,}117$ less objects. Our MC\textsuperscript{2} algorithm is applied at both the partitioning and lifting stages for these codes, demonstrating the algorithm effectiveness and the remarkable gains it can achieve at both stages of the code design.
    \item For $\left(4,8\right)$ UTS objects, SC Code~6 achieves a $1.92\%$ reduction, corresponding to $24{,}000$ fewer objects than the best literature counterpart. This MC\textsuperscript{2} result is reached at a notably less runtime compared with the respective literature technique.
    \item For SC Code~5, a girth-10 code, the MC\textsuperscript{2} algorithm eliminated all cycles of length $\leq 8$, while the code rate is kept at a moderate $0.50$.
\end{itemize}

These results highlight the versatility of the proposed MC\textsuperscript{2} method in reducing various types of harmful structures, as well as its potential as a robust FL optimizer capable of minimizing more complex and detrimental objects beyond short cycles. They also demonstrate its applicability across different design stages and for various types of codes.

\smallskip

\begin{table} [t]
\caption{Evaluation Counts at Partitioning and Lifting Stages (Code Parameters Are Listed in Table~\ref{table_1})}
\vspace{-0.7em}
\centering
\scalebox{1.00}{
\begin{tabular}{|c|c|c|c|c|}
\hline
\multirow{2}{*}{Code name} & \multicolumn{2}{c|}{Partitioning} & \multicolumn{2}{c|}{Lifting} \\
\cline{2-5}
          & Best result & Matched & Best result & Matched  \\
\hline
\cite[SC Code~3]{oocpo} & \multicolumn{2}{c|}{$-$} & \multicolumn{2}{c|}{$8{,}519{,}142$} \\
\hline
SC Code~1 & \multicolumn{2}{c|}{$-$} & $3{,}445{,}562$ & $1{,}872{,}196$ \\
\thickhline
\cite[SC Code~7]{oocpo} & \multicolumn{2}{c|}{$-$} & \multicolumn{2}{c|}{$1{,}586{,}051$} \\
\hline
SC Code~2 & \multicolumn{2}{c|}{$-$} & \multicolumn{2}{c|}{$561{,}162$} \\
\thickhline
\cite[$\left(3,17\right)$]{GRADE} & \multicolumn{2}{c|}{$3{,}060$} & \multicolumn{2}{c|}{$12{,}630$} \\
\hline
SC Code~3 & $218{,}490$ & $5{,}304$ & $2{,}663{,}153$ & $386$ \\
\thickhline
\cite[$\left(4,29\right)$]{GRADE} & \multicolumn{2}{c|}{$11{,}020$} & \multicolumn{2}{c|}{$71{,}085{,}008$} \\
\hline
SC Code~4 & $203{,}009$ & $2{,}855$ & $17{,}282{,}611$ & $41{,}515$ \\
\thickhline
SC Code~5 & \multicolumn{2}{c|}{$33{,}369$} & \multicolumn{2}{c|}{$7{,}613$} \\
\thickhline
\cite[Code 10]{channel_aware} & \multicolumn{2}{c|}{$-$} & \multicolumn{2}{c|}{$54{,}350{,}669$} \\
\hline
SC Code~6 & \multicolumn{2}{c|}{$-$} & $520{,}090$ & $85{,}816$ \\
\thickhline
\cite[$\left(4,17\right)$]{GRADE} & \multicolumn{2}{c|}{$-$} & \multicolumn{2}{c|}{$34{,}442{,}586$} \\
\hline
TC Code~1 & \multicolumn{2}{c|}{$-$} & $75{,}822{,}980$ & $1{,}312{,}451$ \\
\hline
\end{tabular}}
\label{table_4}
\vspace{-1.2em}
\end{table}

\item \textbf{Estimation accuracy:} 
Our fitting approach demonstrates high precision in estimating minimum object counts the MC\textsuperscript{2} algorithm can reach, with deviations between estimated and actual results being consistently low. In particular, such deviations range from $3.59\%$ to $6.23\%$, except for an isolated case with TC Code~1, where the deviation is significantly higher at \textbf{$46.83\%$}. For instance, the estimated cycle-$6$ count for SC Code~1 is $13{,}399$, whereas the actual count is $12{,}937$, resulting in a deviation of just $3.57\%$. This accuracy highlights the reliability of our predictions regarding the resulting code. It should also be noted that these estimation results are obtained with far fewer iterations than what the minimization algorithm uses, and yet, they still provide valuable insights into the algorithm behaviour.

\begin{table} [h]
\caption{Run-times of the Algorithms (Code Parameters Are Listed in Table~\ref{table_1})}
\vspace{-0.5em}
\centering
\scalebox{1.00}{
\begin{tabular}{|c|c|c|c|}
\hline
\multirow{2}{*}{Code name} & \multirow{2}{*}{Design stage} & \multicolumn{2}{c|}{Run-time (in minutes)} \\
\cline{3-4}
      &    & Best result & Matched \\
\hline
\cite[SC Code~3]{oocpo} & Lifting & \multicolumn{2}{c|}{$70.03$} \\
\hline
SC Code~1 & Lifting & \,\; $0.20$ \,\; & $0.12$ \\
\thickhline
\cite[$\left(4,29\right)$]{GRADE} & Lifting & \multicolumn{2}{c|}{$293.409$} \\
\hline
SC Code~4 & Lifting & $10.20$ & $0.15$ \\
\hline
\end{tabular}}
\label{table_5}
\vspace{-0.7em}
\end{table}

\item \textbf{Computational complexity:} 
The MC\textsuperscript{2} algorithm consistently achieves  superior results (in the overwhelming majority of cases) or comparable results (in very few cases) in minimizing object counts. Typically, to reach the same literature count, our MC\textsuperscript{2} algorithm offers orders of magnitude time-complexity reduction. Moreover, even to notably overcome the literature result, the MC\textsuperscript{2} algorithm typically offers at most the same complexity, with various cases exhibiting remarkable complexity reduction as well. Key observations include:

\begin{itemize}
    \item For SC Code~1, SC Code~2, SC Code~4 (lifting stage only for this code), and SC Code~7, our algorithm shows better complexity outcomes, since it requires $0.4$ to $2.0$ orders of magnitude fewer evaluations to reach its minimum (which is a better result) compared with literature methods, and $0.45$ to $3.23$ orders of magnitude fewer evaluations to replicate the literature results.
    \item For SC Code~3 (lifting stage only), SC Code~4 (partitioning stage only), and TC Code~1, our algorithm reaches the best results of literature techniques with $0.59$ to $1.51$ orders of magnitude fewer evaluations. Achieving its optimal results may require additional number of evaluations, making its complexity comparable to or slightly higher than literature methods for these specific codes, but with notably reduced object counts.
    \item For SC Code~3 (partitioning stage only), our algorithm requires higher number of evaluations of objective function to achieve both its own best result and the literature best. However, the resulting code has significantly fewer objects with only a modest increase in complexity. It should be noted that this is the only case we encountered in which achieving the best result in literature requires a minor complexity increase by our method.
\end{itemize}

Additionally, the MC\textsuperscript{2} algorithm is highly suitable for parallelization, since the most computationally demanding task (i.e., evaluation of objective function) can be distributed across multiple threads. Efficiency of this approach, in terms of the required run-time, is evident in Table~\ref{table_5}. For both SC Code~1 and SC Code~4, our method achieves significantly better results in terms of run-time compared with the methods reported in the literature, both to match and even to outperform the literature counts.

\smallskip

\item \textbf{Improved simulation results:}
\begin{itemize}
\item For SC Code~3, the literature code exhibits an error floor, whereas our code does not, thanks to the reduced number of detrimental objects, resulting in up to $1.46$ orders of magnitude improvement in FER at a channel erasure rate of $0.1925$. Furthermore, we collected and analyzed the frames with uncorrected erasures of the literature code in the error-floor region, and observed that a dominant size-$10$ stopping set is responsible for the majority of decoding failures.\footnote{A size-$a$ stopping set is a configuration with $a$ VNs and with CNs each of degrees at least $2$. Stopping sets are the main cause of decoding failures for the peeling decoder~\cite{banihashemi3}.} This configuration is shown in Fig.~\ref{fig4}, and it can be observed that it encompasses five cycles-$8$. Our code does not contain such a detrimental configuration, demonstrating that the MC\textsuperscript{2} optimization we performed to reduce the number of cycles-$8$ significantly improved our code performance, resulting in no visible error floor.

\item SC Code~4 consistently outperforms its counterpart across the entire SNR range of interest, particularly in the waterfall region, due to the reduced number of detrimental objects, which leads to less message-passing dependencies. In this case, the maximum observed gain in FER reaches $1.61$ orders of magnitude at a signal-to-noise ratio (SNR) of $2.825$~dB. We also observe that at $10^{-6}$ FER level, our code has about $0.05$~dB gain compared with its literature counterpart.
\end{itemize}
These results confirm that the FL optimizations achieved by the proposed MC\textsuperscript{2} method leads to substantial practical performance improvements across different channels and decoding algorithms compared with available methods.

\begin{figure}
    \centering
    \includegraphics[width=0.30\textwidth]{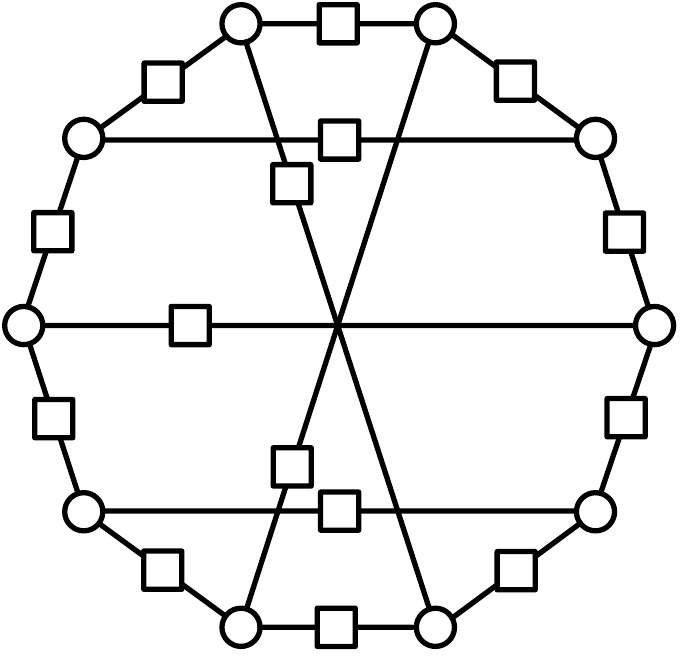}
    \caption{A size-10 stopping set in the Tanner graph of the literature counterpart of SC Code~3, labeled as \cite[$(3,17)$]{GRADE}, with parameters $\left(\gamma, \kappa, z, L, m\right) = \left(3, 17, 7, 10, 9\right)$. This configuration is identified as the main reason behind the faulty frames in the error-floor region.}
    \label{fig4}
\vspace{-1.0em}
\end{figure}

\smallskip

\item \textbf{Ergodicity and number of recurrent classes of the Markov chain:}
Although it is not formally proven, our optimization problem is highly unlikely to be convex. Consequently, the Markov chain constructed for the algorithm is unlikely to possess a single recurrent class. Instead, it is expected to have multiple recurrent classes. From our fitting and estimation results, we observed that the cardinality of the recurrent classes where data is collected is in the order of $O(N^k)$, where $k$ is smaller than $\gamma \kappa$ but larger than $d$, the number of entries updated at each transition. This empirical observation, shown in Table~\ref{table_3}, suggests that the Markov chains are non-ergodic and possess multiple recurrent classes each. Nevertheless, it has been shown that the observed recurrent class during the algorithm, considered as a small chain, satisfies the convergence condition of its main distribution. Furthermore, based on Theorem~\ref{theorem:convergence}, it is expected that the optimal solution reached by the algorithm will be arbitrarily close to the optimal solution of this small chain under analysis.

\smallskip

\item \textbf{Suboptimal recurrent classes and compatibility with a GD-based probabilistic approach:}
The performance of the MC\textsuperscript{2} minimizer heavily depends on the recurrent class it enters, which influences the size of the search space, computational complexity, and the closeness of the reached minimum to the global minimum. Proper input vector initialization is therefore crucial. GD-based initialization proves highly effective, directing the optimizer toward recurrent classes with better local minima. As shown in Table~\ref{table_6}, GD-initialized runs achieve significantly lower cycle counts compared with those initialized with a uniform distribution (UNF). As intriguing, the performance gap is wider in favor of our algorithm in such situations compared with corresponding GD-based method in the literature \cite{GRADE}. In the case of lifting, a practical strategy is to perform multiple random initializations, execute the algorithm for a limited number of iterations in each run, and then proceed with the best-performing result to reduce the risk of being trapped in a poor recurrent class (poor local optimum).

\begin{table}
\caption{Resulting Cycle-$8$ Counts of Codes Initialized with Uniform or Gradient-Descent-Based Distribution at Partitioning Stage (Code Parameters Are Listed in Table~\ref{table_1})}
\vspace{-0.5em}
\centering
\scalebox{1.00}
{
\begin{tabular}{|c|c|c|}
\hline
\makecell{Code name} & \makecell{GD partitioning} & \makecell{UNF partitioning} \\
\hline
\cite[$\left(3,17\right)$]{GRADE} & $13{,}860$ & $24{,}304$ \\
\hline
SC Code~3 & $12{,}530$ & $28{,}567$ \\
\thickhline
\cite[$\left(4,29\right)$]{GRADE} & $528{,}090$ & $1{,}087{,}268$ \\
\hline
SC Code~4 & $ 409{,}973 $ & $1{,}191{,}059$ \\
\hline
\end{tabular}}
\label{table_6}
\vspace{-0.7em}
\end{table}

\smallskip

\item \textbf{Generality of the MC\textsuperscript{2} method as a finite-length optimizer:} 
Although this paper focuses on SC/TC code design during the partitioning and lifting stages, the MC\textsuperscript{2}-based approach can be readily adapted to other finite-length optimization problems in LDPC code design or, more broadly, in coding theory.

This method can be generalized for lifting arrangements of any type of LDPC codes beyond SC codes, applied to MD code design at the relocation stage, or even extended to certain cases of non-binary code design involving edge-weight arrangements. In addition to its minimization algorithm, the estimation and data-fitting components of the MC\textsuperscript{2} approach can also be leveraged in various areas of code design, offering valuable insights into the theoretical aspects of the algorithm.

\end{enumerate}

\section{Conclusion} \label{sec:conclude}

Finite-length (FL) optimization in spatially-coupled (SC) code design is widely considered a computationally-taxing task, especially as more degrees of freedom become available. This work proposes a Markov chain Monte Carlo (MC\textsuperscript{2})-based framework to address this discrete optimization problem, targeting both the partitioning and lifting stages to minimize the number of short cycles and common subgraphs of dominant absorbing sets. The method adopts a probabilistic approach to efficiently search for optimal or near-optimal solutions, while also enabling statistical analysis to estimate the resulting object counts, which is an interesting problem to tackle in its own right. Our numerical results demonstrate that the MC\textsuperscript{2} approach consistently outperforms existing methods in terms of object count reduction, which reaches $48\%$. Our approach also offers superior computational efficiency, which reaches $3.32$ orders of magnitude complexity reduction, compared with the literature in the vast majority of cases examined. Our experimental findings demonstrate that the FL optimization achieved by the MC\textsuperscript{2} method results in notable improvements in practical performance across a range of channels and decoders. In particular, simulation results reveal frame error/erasure-rate gains of up to $1.61$ orders of magnitude compared with literature methods. Moreover, a general consistency is observed between the estimated and actual outcomes. Integrating this approach with complementary probabilistic frameworks, such as the gradient-descent algorithm in \cite{GRADE}, results in additional performance gains. The proposed framework is expected to be applicable beyond SC code design, offering potential solutions for a broader range of discrete optimization problems in LDPC code design, and more generally, in coding theory. Future work includes applying the method to multi-dimensional SC code design as well as targeting more advanced and more detrimental structures, such as absorbing and trapping sets, to further enhance code performance.


\end{document}